\documentclass[12pt]{article}
\usepackage[T1]{fontenc}
\usepackage[dvips]{graphicx}
\graphicspath{{images/}}
\usepackage{verbatim}
\usepackage{amsmath,amsthm}
\usepackage{xspace}
\usepackage{pifont}
\usepackage{graphicx}
\usepackage{amssymb}
\usepackage{epic, eepic}
\usepackage{dsfont}
\usepackage{amssymb}
\usepackage{makeidx}
\usepackage{mathrsfs}
\usepackage{exscale}
\usepackage{color} 
\usepackage{overpic} 
\usepackage{bm}
\usepackage{bbm}
\usepackage{booktabs} 
\usepackage{color, colortbl}
\usepackage{subcaption}
\usepackage[numbers]{natbib}

\RequirePackage[colorlinks,citecolor=blue,urlcolor=blue]{hyperref}

\definecolor{Gray}{gray}{0.9}

\usepackage{amsmath,afterpage}
\usepackage{epsf}
\usepackage{graphics,color}

\def\0{\mathbf{0}}

\def\lam{\lambda}

\def \< {\langle}
\def \> {\rangle}

\def\beqa{\begin{eqnarray}}
\def\eeqa{\end{eqnarray}}
\def\beqas{\begin{eqnarray*}}
\def\eeqas{\end{eqnarray*}}

\newtheorem{theorem}{Theorem}[section]
\newtheorem{lemma}[theorem]{Lemma}

\newtheorem{corollary}[theorem]{Corollary}

\newtheorem{remark}[theorem]{Remark}

\newtheorem{assumption}[theorem]{Assumption}
\numberwithin{equation}{section}
\newcommand{\hatd}[1]{{}}

\newcommand{\bd}{\begin{displaymath}}
\newcommand{\ed}{\end{displaymath}}
\newcommand{\be}{\begin{equation}}
\newcommand{\ee}{\end{equation}}
\newcommand{\bq}{\begin{eqnarray}}
\newcommand{\eq}{\end{eqnarray}}
\newcommand{\bn}{\begin{eqnarray*}}
\newcommand{\en}{\end{eqnarray*}}

\newcommand{\re}{\mathds{R}}

\def\wt{\widetilde}

\def\P{\mathbb{P}}
\def\E{{\mathbb{E}}}

\usepackage{authblk}

\title{Optimal Signal-Adaptive Trading with Temporary and Transient Price Impact}
\author[1]{Eyal Neuman }
\author[2]{Moritz Vo\ss  }
\affil[1]{Department of Mathematics, Imperial College London}
\affil[2]{Department of Mathematics,
University of California, Los Angeles}
 
\begin{document}

 \vspace{-0.5cm}
\maketitle

\begin{abstract}
We study optimal liquidation in the presence of linear temporary and transient price impact along with taking into account a general price predicting finite-variation signal. We formulate this problem as minimization of a cost-risk functional over a class of absolutely continuous and signal-adaptive strategies. The stochastic control problem is solved by following a probabilistic and convex analytic approach. We show that the optimal trading strategy is given by a system of four coupled forward-backward SDEs, which can be solved explicitly. Our results reveal how the induced transient price distortion provides together with the predictive signal an additional predictor about future price changes. As a consequence, the optimal signal-adaptive trading rate trades off exploiting the predictive signal against incurring the transient displacement of the execution price from its unaffected level. This answers an open question from~\citet{Lehalle-Neum18} as we show how to derive  the unique optimal signal-adaptive liquidation strategy when price impact is not only temporary but also transient.  
\end{abstract} 


\begin{description}
\item[Mathematics Subject Classification (2010):] 93E20, 60H30, 91G80
\item[JEL Classification:] C02, C61, G11
\item[Keywords:] optimal portfolio liquidation, price impact,
  optimal stochastic control, predictive signals
\end{description}

\bigskip

\section{Introduction}

The trading costs of executing large orders on an electronic trading platform often arise from the notion of price impact. Price impact refers to the empirical fact that the execution of a large order affects the risky asset's price in an adverse manner leading to less favorable prices. This typically induces additional execution costs for the trader. As a result, a trader who wishes to minimize her trading costs due to price impact has to split her order into a sequence of smaller orders which are executed over a finite time horizon. At the same time, the trader also has an incentive to execute these split orders rapidly because she does not want to carry the risk of an adverse price move far away from her initial decision price. This trade-off between price impact and market risk is usually translated into a stochastic optimal control problem where the trader aims to minimize a risk-cost functional over a suitable class of execution strategies. The corresponding optimal order execution problem has been extensively studied in the literature and continues to be of ongoing interest in research and practice. We refer to the monographs~\cite{cartea15book}, \cite{Gueant:16}, \cite{citeulike:12047995}, as well as the survey papers~\cite{GatheralSchiedSurvey} and~\cite{GokayRochSoner} for a thorough account for the developed price impact models.



In practice, apart from focusing on the aforementioned trade-off between price impact and market risk, many traders and trading algorithms also strive for using short term price predictors in their dynamic order execution schedules. Most of such documented predictors relate to orderbook dynamics as discussed in~\cite{leh16moun, Lehalle-Neum18, citeulike:12820703, cont14}. An example of such price predicting indicator is the order book imbalance signal, measuring the imbalance of the current liquidity in the limit order book; see, e.g., Section 4 of \cite{Lehalle-Neum18} and references therein. Another signal which was studied in the literature in the context of optimal order execution is the order flow imbalance; we refer to~\cite{Car-Jiam-2016, cont14, ludk14, ludk17} and references therein. More examples of trading signals being used in practice can be found in a presentation by~\citet{almg-pres}.


Consequently, one of the main challenges in the area of optimal
trading with price impact deals with the question of how to
incorporate short term predictive signals into a stochastic control
framework of cost-risk minimization. Among the first to address this
issue were~\citet{Car-Jiam-2016} who showed how to account for a
Markovian signal in an optimal execution problem in the presence of
linear temporary and permanent price impact of~\citet{OPTEXECAC00}
type. Their framework was then further generalized
by~\citet{Lehalle-Neum18}, Section 3, and by~\citet{BMO:19} who also
allowed for non-Markovian finite variation signals; see
also~\citet{CasgrainJaimungal:19} for incorporating latent factors
into the modeling framework of~\cite{Car-Jiam-2016}.  Subsequently, a
Markovian signal and transient price impact for a general class of
impact decay kernels as proposed by~\citet{GSS} were first confronted
in~\citet{Lehalle-Neum18}, Section 2. In contrast to purely temporary
price impact, transient impact on execution prices persists and decays
over a certain period of time after each trade. As a consequence,
optimal trading strategies in the presence of purely transient price
impact are typically singular. Indeed, they tend to trigger a
displacement of the market price from its unaffected level via an
instantaneous non-infinitesimal block trade in order to systematically
exploit the successively decaying impact at a finite trading rate as,
e.g., illustrated by the explicit results
in~\citet{Ob-Wan2005}. Mathematically, this renders the analysis of
optimal signal-adaptive trading strategies with transient price impact
much more intricate. A remedy, employed by the authors
in~\cite{Lehalle-Neum18}, consists of confining to deterministic (or
static) strategies which only use information of the predictive signal
at initial time. The important question about existence and
characterization of an optimal signal-adaptive strategy with transient
price impact was left open. In fact, in~\citet{stat-adp18} it was
shown that a strategy which is updated by information from the signal
several times during the liquidation period can significantly improve
the trading performance compared to an optimal static strategy. A
partial solution to this problem was proposed by~\citet{LorenzSchied},
who considered an execution model with exponentially 
decaying transient price impact, but without including a risk-aversion
term in the cost functional. Under the assumption that the signal is
absolutely continuous with a square integrable derivative the
optimal adapted strategy was derived.  Moreover, since only transient
price impact was considered in the model of \cite{LorenzSchied}, the
optimal strategy is singular and involves block trades. We will show
in this paper that once assuming that the price impact is both
transient and temporary, we can omit these regularity assumptions on
the signal and obtain absolutely continuous optimal trading
strategies. In addition, we show that the influence of the
risk-aversion term in the cost functional changes drastically the
qualitative behaviour of the optimal trading speed.

The main result in the present paper gives an answer to the open
question from~\cite{Lehalle-Neum18}. Specifically, in order to
optimize trading costs in the presence of exponentially decaying
transient price impact as proposed by~\citet{Ob-Wan2005} over a
sensible set of strategies adapted to the signal's filtration, we
adopt the price impact model from~\citet{GARLEANU16}. We incorporate
into the trader's cost-risk functional besides a linear transient
price impact component \`a la Obizhaeva-Wang also a linear temporary
price impact component of Almgren-Chriss type. This unifying framework
with temporary and transient price impact quadratically smoothens the
problem and rules out singular optimizers by naturally constraining
strategies to be absolutely continuous. Moreover, it turns out that
the probabilistic and convex analytic calculus of variations approach
from~\citet{BankSonerVoss:17} can be brought to bear to compute
explicitly optimal signal-adaptive strategies, also in a setup which
allows for more general non-Markovian signals compared
to~\cite{Lehalle-Neum18, GARLEANU16, LorenzSchied}. Following the
analysis in~\cite{GARLEANU16}, the crucial idea is to introduce the
displacement of the execution price from its unaffected level due to
transient price impact as an additional state variable. Then, similar
to~\cite{BankSonerVoss:17} the optimal control is characterized by a
system of coupled linear forward-backward stochastic differential
equations (FBSDEs) which is augmented by a linear forward equation for
the transient price distortion as well as an associated adjoint linear
backward SDE. This linear system can be decoupled and solved in closed
form. Its solution provides an explicit description of the optimal
trading rate. 
It turns out that the transient price distortion provides together
with the predictive signal an additional predictor for future price
changes. Accordingly, the optimal trading rate compensates the
exploitation of the predictive signal with the incurred transient
price impact.

Our results in this paper improve the results of~\citet{Car-Jiam-2016}
and~\citet{BMO:19} as we additionally allow for transient price
impact. In our setting the optimal strategy depends on the entire
trading trajectory, unlike the Markovian optimal strategies in the
strictly instantaneous price impact case.  Our main results also
generalize the results
of~\citet{GraeweHorst:17},~\citet{ChenHorstTran:19}, \citet{GSS},
\citet{SchiedStrehleZhang} and \citet{Strehle:17}, who study optimal
liquidation with both temporary and transient price impact, but
without a predictive price signal. This class of problems typically
leads to deterministic optimal strategies. However, signal-adaptive
optimal execution schedules have major practical significance, as
described above. Finally, our paper is also related to a very recent
work by~\citet{Forde21}, where an optimal liquidation problem with
power-law transient price impact and Gaussian signals is studied.

Our findings also relate to a class of optimal portfolio choice
problems: see, e.g.,~\citet{GARLEANU16}, and~\citet{Ekren-Karbe19} for
a setup of a portfolio optimizing agent that tries to exploit
partially predictable returns while facing linear temporary and
transient price impact. Unlike in our optimal execution framework, the
trading time horizon in these optimal investment problems is infinite,
and the agent perpetually invests simultaneously in a few assets
having their own signals. The ansatz for the value function is
typically a second order polynomial, which makes the derivation of the
latter as well as the corresponding optimal strategy much easier than
in the parabolic case where the time horizon is finite.  Moreover,
both~\cite{GARLEANU16} and~\cite{Ekren-Karbe19} study a Markovian
setup via dynamic programming techniques: \cite{GARLEANU16} describes
the optimal trading rate where the predictable returns are driven by a
jump-diffusion process, an assumption which is not needed in the
present paper since we allow a general signal in our model;
\cite{Ekren-Karbe19} studies the case where the signal is a Markovian
diffusion process which interacts with the asset prices through their
drift vector and covariance matrix. They derive an asymptotic optimal
trading strategy in the case when both temporary and transient price
impact tend to zero. In contrast, we derive the optimal strategy not
under the restriction of vanishing price impact.
  

The rest of the paper is organized as follows. In Section \ref{sec:setup} we introduce our optimal execution problem with temporary and transient price impact and predictable finite-variation signal. Our main result, an explicit solution to our optimal stochastic control problem, is presented in Section~\ref{sec:main}. Section~\ref{sec:illustration} contains some illustrations. The technical proofs are deferred to Section~\ref{sec:proof}.

\section{Model setup and problem formulation} \label{sec:setup}

Motivated by~\citet{Lehalle-Neum18} we introduce in the following a variant of the optimal signal-adaptive trading problem with transient price impact which was studied in Section 2 therein. 

Let $T>0$ denote a finite deterministic time horizon and fix a filtered probability space $(\Omega, \mathcal F,(\mathcal F_t)_{0 \leq t\leq T}, \P  )$ satisfying the usual conditions of right continuity and completeness. The set $\mathcal H^2$ represents the class of all (special) semimartingales~$P=(P_t)_{0 \leq t\leq T}$ whose canonical decomposition $P = \bar M + A$ into a (local) martingale $\bar M=(\bar M_t)_{0 \leq t\leq T}$ and a predictable finite-variation process $A=(A_t)_{0 \leq t\leq T}$ satisfies
\be \label{ass:P} 
E \left[ \langle  \bar M  \rangle_T \right] + E\left[\left( \int_0^T |dA_s| \right)^2 \right] < \infty.  
\ee

We consider a trader with an initial position of $x>0$ shares in a risky asset. The number of shares the trader holds at time $t\in [0,T]$ is prescribed as 
\begin{equation} \label{def:X}
X^u_t \triangleq x-\int_0^tu_sds
\end{equation}
where $(u_s)_{s \in [0,T]}$ denotes her selling rate which she chooses from a set of strategies
\be \label{def:admissset} 
\mathcal A \triangleq \left\{ u \, : \, u \textrm{ progressively measurable s.t. } \mathbb E\left[ \int_0^Tu_s^2 ds \right] <\infty \right\}.
\ee
We assume that the trader's trading activity causes price impact on the risky asset's execution price in the sense that her orders are filled at prices
\be \label{def:S}
S_{t} \triangleq P_{t} - \lam u_t - \kappa Y^u_t \qquad (0 \leq t \leq T), 
\ee
where $P$ denotes some unaffected price process in $\mathcal H^2$ and
\be \label{def:Y} 
Y^u_t \triangleq e^{-\rho t} y + \gamma\int_0^t e^{-\rho (t-s)}  u_s ds \qquad (0 \leq t \leq T) 
\ee
with some $y > 0$. Specifically, motivated by \citet{GARLEANU16} the trader's trading not only instantaneously affects the execution price in~\eqref{def:S} in an adverse manner through linear temporary price impact $\lambda > 0$ \`a la~\citet{OPTEXECAC00}; it also induces a longer lasting price distortion $Y^u$ because of linear transient price impact $\kappa > 0$ and $\gamma > 0$ as proposed by~\citet{Ob-Wan2005}. We assume that the transient price impact, which starts from an initial value $y>0$, persists and  decays only gradually over time at some exponential resilience rate $\rho > 0$. Also note that the unaffected price process $P \in \mathcal H^2$ includes a general signal process $A$ which is observed by the trader.

We now suppose that the trader's optimal trading objective is to unwind her initial position $x>0$ in the presence of temporary and transient price impact, along with taking into account the asset's general price signal $A$, through maximizing the performance functional 
\begin{equation} \label{def:objective}
\begin{aligned}
    & J(u) \triangleq \mathbb{E} \Bigg[ \int_0^T (P_t - \kappa Y^u_t) u_t dt - \lambda \int_0^T u^2_t dt + X_T^u P_T \Bigg.  \\ & \Bigg. \hspace{60pt} -\phi \int_0^T (X_t^u)^2 dt - \varrho (X_T^u)^2 \Bigg]
\end{aligned}
\end{equation}
via her selling rate $u \in \mathcal A$. The interpretation is as follows (cf. also Remark~\ref{rem:problem}.1.) below). The first three terms in~\eqref{def:objective} represent the trader's terminal wealth; that is, her final cash position including the accrued trading costs which are induced by temporary and transient price impact as prescribed in~\eqref{def:S}, as well as her remaining final risky asset position's book value. The fourth and fifth terms in~\eqref{def:objective} implement a penalty $\phi > 0$ and $\varrho > 0$ on her running and terminal inventory, respectively. Also observe that $J(u) < \infty$ for any admissible strategy $u \in \mathcal A$.

Our main goal in this paper is to solve the corresponding optimal stochastic control problem
\begin{equation} \label{def:optimization}
J(u) \rightarrow    \max_{u \in \mathcal A}.
\end{equation}

\begin{remark} \label{rem:problem}
\begin{enumerate}
    \item In case of purely temporary price impact, i.e., $\kappa = 0$, the performance functional in~\eqref{def:objective} and the associated optimization problem in~\eqref{def:optimization} is quite standard in the literature on optimal trading and execution problems. It was first introduced by~\cite{AlmgrenSIFIN, Forsythetal} and then subsequently studied, e.g., in~\cite{SchiedFuel, AnkirchnerJeanblancKruse:14, GraeweHorstQiu:15, Car-Jiam-2016, Lehalle-Neum18, BMO:19}.
    
    \item Our problem formulation in~\eqref{def:objective} with temporary and transient price impact is very similar to the framework introduced in~\citet{GARLEANU16} which was then further analyzed by~\citet{Ekren-Karbe19}. In contrast to their setup, we focus on a finite time horizon $T < \infty$. We also allow for more general price signal processes $A$ and not only a linear factor process as in~\cite{GARLEANU16} or a Markovian diffusion-type process as in~\cite{Ekren-Karbe19}. Moreover, we obtain an explicit solution to our optimization problem in~\eqref{def:optimization}, akin to the results established by~\cite{GARLEANU16} in their simpler framework, and do not necessitate an asymptotic analysis as carried out in~\cite{Ekren-Karbe19}.
    
    \item As mentioned at the beginning of this section our framework
      presented above also aims at following up on the optimal
      signal-adaptive trading problem in the presence of transient
      price impact which was studied in Section 2
      of~\cite{Lehalle-Neum18}. Therein, the authors confine
      themselves to analyze only deterministic optimal
      strategies. Indeed, optimal strategies in a purely transient
      price impact setup are typically singular (cf., e.g.,
      \cite{Ob-Wan2005, GSS, LorenzSchied, BankVoss:19}) which renders
      the analysis for signal-adaptive strategies mathematically much
      more intricate. As a remedy, we adopt the approach
      from~\cite{GARLEANU16}. We incorporate into the cost functional
      in~\eqref{def:objective} in addition to the transient price
      impact also a temporary impact component via $\lambda > 0$ which
      rules out singular optimizers. This quadratically smoothens the
      problem and very naturally constrains strategies to be
      absolutely continuous. Also note that we do not require the
      signal processes $A$ to be an integrated Markov process as
      in~\cite{Lehalle-Neum18}.

    \item For $A\equiv 0$, that is, without price signal process, but
      with terminal liquidation constraint $X^u_T = 0$ $\P$-a.s. the
      above optimization problem in~\eqref{def:optimization} was
      studied in~\citet{GraeweHorst:17} allowing for stochastic
      resilience and temporary price impact processes
      $(\rho_t)_{0 \leq t \leq T}$ and
      $(\lambda_t)_{0 \leq t \leq T}$, respectively. For the
      corresponding explicitly available deterministic solution in the
      case of constant coefficients and $\phi =0$ we refer
      to~\cite{ChenHorstTran:19}. As, e.g, in~\cite{Car-Jiam-2016,
        Lehalle-Neum18, BMO:19}, we do not incorporate a terminal
      state constraint in our optimization problem
      in~\eqref{def:optimization}. Note, however, that the terminal
      penalty $\varrho > 0$ on the remaining risky asset position
      allows to virtually enforce a liquidation constraint by choosing
      a large value for~$\varrho$ (see also the illustrations in
      Section~\ref{sec:illustration} below).
\end{enumerate}
\end{remark}

\section{Main result} \label{sec:main}

Our main result is an explicit description of the optimal strategy for problem~\eqref{def:optimization}. To state our result it is convenient to introduce for
\be \label{def:A} 
L \triangleq \begin{pmatrix}
    0 & 0 & -1 & 0 \\ 0 & -\rho & \gamma & 0 \\ -\phi/\lambda & \kappa\rho/(2\lambda) & 0 & \rho/(2\lambda) \\ 0 & 0 & \kappa\gamma & \rho
    \end{pmatrix} \in \mathbb{R}^{4 \times 4} 
\ee
the functions $S(t) = (S_{ij}(t))_{1\leq i,j \leq 4}$ given by the matrix exponential
\be \label{def:matrixExp}
S(t) \triangleq e^{Lt} \qquad (t \geq 0). 
\ee
We further define $G(t) = (G_i(t))_{1 \leq i \leq 4}$ as   
\begin{equation} \label{def:G}
G(t) \triangleq \left( \varrho/\lambda, -\kappa/(2\lambda),-1,0 \right) S(t) \qquad (t\geq 0) 
\end{equation}
 and let 
\begin{equation}
\begin{aligned} \label{def:vs}
v_0(t) \triangleq \left( 1- \frac{G_4(t)}{G_3(t)} \frac{S_{4,3}(t)}{ S_{4,4}(t)} \right)^{-1},  \qquad
&& v_1(t) \triangleq & \;
          \frac{G_4(t)}{G_3(t)}\frac{S_{4,1}(t)}{S_{4,4}(t)}
          -\frac{G_1(t)}{G_3(t)}, \\
v_2(t) \triangleq 
          \frac{G_4(t)}{G_3(t)}\frac{S_{4,2}(t)}{S_{4,4}(t)}
         -\frac{G_2(t)}{G_3(t)},  \qquad
&& v_3(t) \triangleq & \; \frac{G_4(t)}{G_3(t)}
\end{aligned}
\end{equation} 
for all $t \in [0,\infty)$. The functions $(S_{4,j}(t))_{1\leq j \leq 4}$ and $(G_i(t))_{1 \leq i \leq 4}$ can be computed explicitly and are given in~\eqref{eq:S41}-\eqref{eq:S44} and~\eqref{eq:G1}-\eqref{eq:G4}, respectively, in Section~\ref{subsec:matrixexponential} below. Lemma~\ref{lemma-g-s} therein also shows that $v_1(\cdot), v_2(\cdot), v_3(\cdot)$ are well-defined. In addition, we make following assumption below (see also Remark~\ref{rem:main}.1.).  

\begin{assumption} \label{assump:main}
We assume that the set of parameters $\xi \triangleq (\lambda,\gamma, \kappa, \rho, \varrho, \phi, T) \in \re^7_+$ are chosen such that
\be \label{g-s-cond} 
 G_3(T-t) S_{4,4}(T-t) \neq G_4(T-t) S_{4,3}(T-t) \qquad (0 \leq t \leq T). 
 \ee
\end{assumption}

Finally, let $ \mathbb E_{t}$ denote the expectation conditioned on $\mathcal F_t$ for all $t\in [0,T]$. We are now ready to state our main theorem. 

\begin{theorem} \label{thm:main} 
Under Assumption~\ref{assump:main}, there exists a unique optimal strategy $\hat{u} \in \mathcal A$ to problem~\eqref{def:optimization}. It is given in linear feedback form via
\be \label{eq:opt_u} 
\begin{aligned}
\hat{u}_{t} = & \; v_0(T-t) \Bigg( v_1(T-t) X^{\hat{u}}_{t} + v_2(T-t) Y^{\hat{u}}_{t} \Bigg. \\
& + \frac{1}{2\lam} \Bigg. \bigg( v_3(T-t)\mathbb E_t\left[
   \int_t^T  \frac{S_{4,3}(T-s)}{S_{4,4}(T-t)}dA_s \right] - \mathbb E_{t}\left[
   \int_{t}^T\frac{G_3(T-s)}{G_{3}(T-t)}dA_{s} \right] \bigg) \Bigg)
\end{aligned}
\ee
for all $t \in (0,T)$.
\end{theorem}  

The proof of Theorem~\ref{thm:main} is deferred to
Section~\ref{subsec:proofmainthm}. We observe that the optimal trading rate $\hat{u}$ in~\eqref{eq:opt_u} is affine-linear in both the current inventory $X^{\hat{u}}$ as well as the current price distortion $Y^{\hat{u}}$. The affine part comes from the general predictive signal~$A$. In particular, it turns out that the transient price displacement from the unaffected level serves as an additional predictor for future price changes. As a consequence, the optimal trading rate trades off exploiting the predictive signal $A$ against incurring transient price distortion $Y^{\hat{u}}$. These findings generalize the observations made in~\citet{GARLEANU16} for optimal portfolio choice problems with infinite horizon in a Markovian setup. Put differently, the optimal stock holdings $X^{\hat{u}}$ prescribed by the optimal selling rate $\hat{u}$ in~\eqref{eq:opt_u} together with the optimally controlled price distortion $Y^{\hat{u}}$ in~\eqref{def:Y} solve a two-dimensional system of coupled linear (random) ordinary differential equations. Its solution can be computed numerically via the associated fundamental solution. Specifically, introducing the process
\begin{equation} \label{def:zeta}
    \hat{\zeta}_t \triangleq \frac{v_0(T-t)}{2\lam} \bigg( v_3(T-t)\mathbb E_t\left[
   \int_t^T  \frac{S_{4,3}(T-s)}{S_{4,4}(T-t)}dA_s \right] - \mathbb E_{t}\left[
   \int_{t}^T\frac{G_3(T-s)}{G_{3}(T-t)}dA_{s} \right] \bigg)
\end{equation}
as well as the matrix-valued function
\begin{align}
    B(t) \triangleq
    \begin{pmatrix}
    -v_0(T-t) v_1(T-t) &  -v_0(T-t) v_2(T-t) \\
    \gamma v_0(T-t) v_1(T-t) &  \gamma v_0(T-t) v_2(T-t) - \rho
    \end{pmatrix} \label{def:B}
\end{align}
for all $t \in [0,T]$, we obtain following


\begin{corollary} \label{cor:main}
Under Assumption~\ref{assump:main} let $\Phi(t) \in \mathbb R^{2\times 2}$ be the unique nonsingular fundamental solution to the matrix differential equation 
\begin{equation} \label{eq:ODEPhi}
     \Phi(0) = I, \quad \dot{\Phi}(t) = B(t) \Phi(t) \quad (0 \leq t \leq T)
\end{equation}
with identity matrix $I \in \mathbb R^{2\times 2}$ and $B$ as defined in~\eqref{def:B}. Then the optimal stock holdings~$X^{\hat{u}}$ and the corresponding optimally controlled price distortion $Y^{\hat{u}}$ of the optimal strategy $\hat{u}$ from Theorem~\ref{thm:main} are given by      
\begin{equation} \label{def:optXY}
\begin{pmatrix}
X^{\hat{u}}_t \\
Y^{\hat{u}}_t
\end{pmatrix}
= \Phi(t) \left(
\begin{pmatrix}
x \\
y
\end{pmatrix}
+ \int_0^t \hat{\zeta}_s \, \Phi^{-1}(s) \, b \, ds 
\right),
\end{equation}
where $b \triangleq (-1, \gamma)^\top \in \mathbb R^2$.
\end{corollary}

Again, the proof of Corollary~\ref{cor:main} can be found in Section~\ref{subsec:proofmainthm}. 

\begin{remark} \label{rem:main}
\begin{enumerate}
\item Assumption~\ref{assump:main} merely ensures that $v_0(\cdot)$ in~\eqref{def:vs} is well-defined. In fact, showing that~\eqref{g-s-cond} holds for any values of parameters $\xi \in \re^7_+$ seems intractable. However, given a set of parameters $\xi$ and verifying that~\eqref{g-s-cond} is satisfied is an easy task by using the explicit formulas for $S_{4,3},S_{4,4},G_3,G_4$ in~\eqref{eq:S43}, \eqref{eq:S44}, \eqref{eq:G3}, \eqref{eq:G4}. We numerically checked this for all $\xi \in [0,100]^7$, which includes all reasonable values of parameters.  

\item The special case where $\kappa = 0$ in the performance functional in~\eqref{def:objective}, i.e., considering only temporary price impact, corresponds to~\citet{BMO:19}, and~\citet{Lehalle-Neum18}, Section~3. One can check with the explicit expressions from Section~\ref{subsec:matrixexponential} that our result in Theorem~\ref{thm:main} retrieves the optimal solution from~\cite{BMO:19}, Theorem 3.1, as well as, in a Markovian setting, from~\cite{Lehalle-Neum18}, Proposition~3.2, in the limiting case when $\kappa$ tends to zero.

\item Note that our optimal strategy in Theorem~\ref{thm:main} is adapted to the underlying filtration $(\mathcal F_t)_{0 \leq t\leq T}$ and hence steadily updates its information about the price signal process $A$. This is in stark  contrast to the signal-adaptive optimal trading framework with transient price impact studied in~\citet{Lehalle-Neum18}, Section 2, where strategies are confined to be static (i.e., deterministic), taking only the information of the price signal at initial time 0 into account. 
\end{enumerate}
\end{remark}

\section{Illustration} \label{sec:illustration}

Similar to \citet{Lehalle-Neum18} we will illustrate in this section our main result in the special case where the signal process $A$ is given by 
\be \label{a-i}
A_t = \int_0^t I_s ds \qquad (t \geq 0)  
\ee
with $I=(I_t)_{t \geq 0}$ following an autonomous Ornstein-Uhlenbeck process with dynamics
\be 
\begin{aligned} \label{I-OU} 
I_{0}&=\iota, \quad dI_{t} = -\beta I_{t}\, dt +\sigma \, dW_{t} \qquad (t\geq 0).
\end{aligned} 
\ee
Here, $W=(W_t)_{t \geq 0}$ denotes a standard Brownian motion which is defined on our underlying filtered probability space and $\beta, \sigma>0$ are some constants. Having at hand our general result from Theorem~\ref{thm:main} we immediately obtain following optimal trading strategy in this case.

\begin{corollary} \label{cor:OU} 
Assume that the signal process $A$ is given by~\eqref{a-i}. Then the unique optimal trading rate $\hat{u} \in \mathcal A$ from Theorem~\ref{thm:main} simplifies to
\begin{align}
\hat{u}_{t} = & \; v_0(T-t) \Bigg( v_1(T-t) X^{\hat{u}}_{t} + v_2(T-t) Y^{\hat{u}}_{t} \Bigg.  \label{u-sol2} \\
& + \frac{I_t}{2\lam} \Bigg. \bigg( v_3(T-t)  
   \int_t^T e^{-\beta(s-t)} \frac{S_{4,3}(T-s)}{S_{4,4}(T-t)} ds - \int_t^T e^{-\beta(s-t)}  \frac{G_3(T-s)}{G_{3}(T-t)}ds \bigg) \Bigg).  \nonumber 
\end{align}
for all $t \in (0,T)$.
\end{corollary}  

\begin{remark}
Observe that the optimal trading rate in~\eqref{u-sol2} is signal-adaptive, i.e., adapted to the filtration generated by $I$, in contrast to the optimal solution presented in Section 2.3 of~\cite{Lehalle-Neum18}.   
\end{remark}

In Figures~\ref{fig:example1} to~\ref{fig:example3} we plot the
signal-adaptive optimal liquidation inventory
$\hat{X} \triangleq (X^{\hat{u}}_t)_{0 \leq t \leq T}$ with initial
position $x=10$ along with the corresponding optimal selling rate
$(\hat{u})_{0 \leq t \leq T}$ and optimally controlled price
distortion $\hat{Y} \triangleq (Y^{\hat{u}}_t)_{0 \leq t \leq T}$ with
$y=0$ obtained from Corollary~\ref{cor:OU} (by using also
Corollary~\ref{cor:main}) for three different realisations of the
signal process $(A_t)_{0 \leq t \leq T}$ in~\eqref{a-i}. The trader's
planning horizon is $T=10$. As for the model parameters, we fix the
values
\begin{equation} \label{model-par}
    \kappa = 1 ,\quad \gamma = 1,\quad \rho = 1, \quad \lambda = 0.5, \quad \phi = 0.1, \quad \varrho = 10,
\end{equation}
as well as
\begin{equation} \label{signal-par}
    \iota = 1, \quad \beta = 0.1,\quad \sigma = 0.5, 
\end{equation}
similar to the parameters in~\cite{Lehalle-Neum18} (cf. also the
empirical analysis in Section~4.2 therein). As mentioned in
Remark~\ref{rem:main}.1. above, one can easily check that
Assumption~\ref{assump:main} holds true with this set of parameter
values. We also compare graphically our signal-adaptive optimal
liquidation strategy $\hat{X}$ with (i) the inventory
$\tilde{X} \triangleq (\tilde{X}_t)_{0 \leq t \leq T}$ and
corresponding induced price distortion trajectory
$\tilde{Y} \triangleq (\tilde{Y}_t)_{0 \leq t \leq T}$ which ignores
the price signal, i.e., $I \equiv 0$ in~\eqref{u-sol2}; (ii) the
optimal signal-adaptive inventory trajectory
$\bar{X} \triangleq (\bar{X}_t)_{0 \leq t \leq T}$ for the purely
temporary price impact case from Theorem 3.1 in~\cite{BMO:19}. Note
that the former is simply the optimal strategy for the maximization
problem in~\eqref{def:optimization} where the trader presumes that the
unaffected price process $P$ has no signal, and the latter corresponds
to the optimal strategy where $\kappa = 0$ in \eqref{def:objective},
i.e., the trader ignores transient price distortion.

\begin{figure}[htbp!]
  \begin{center}
    \includegraphics[scale=.5]{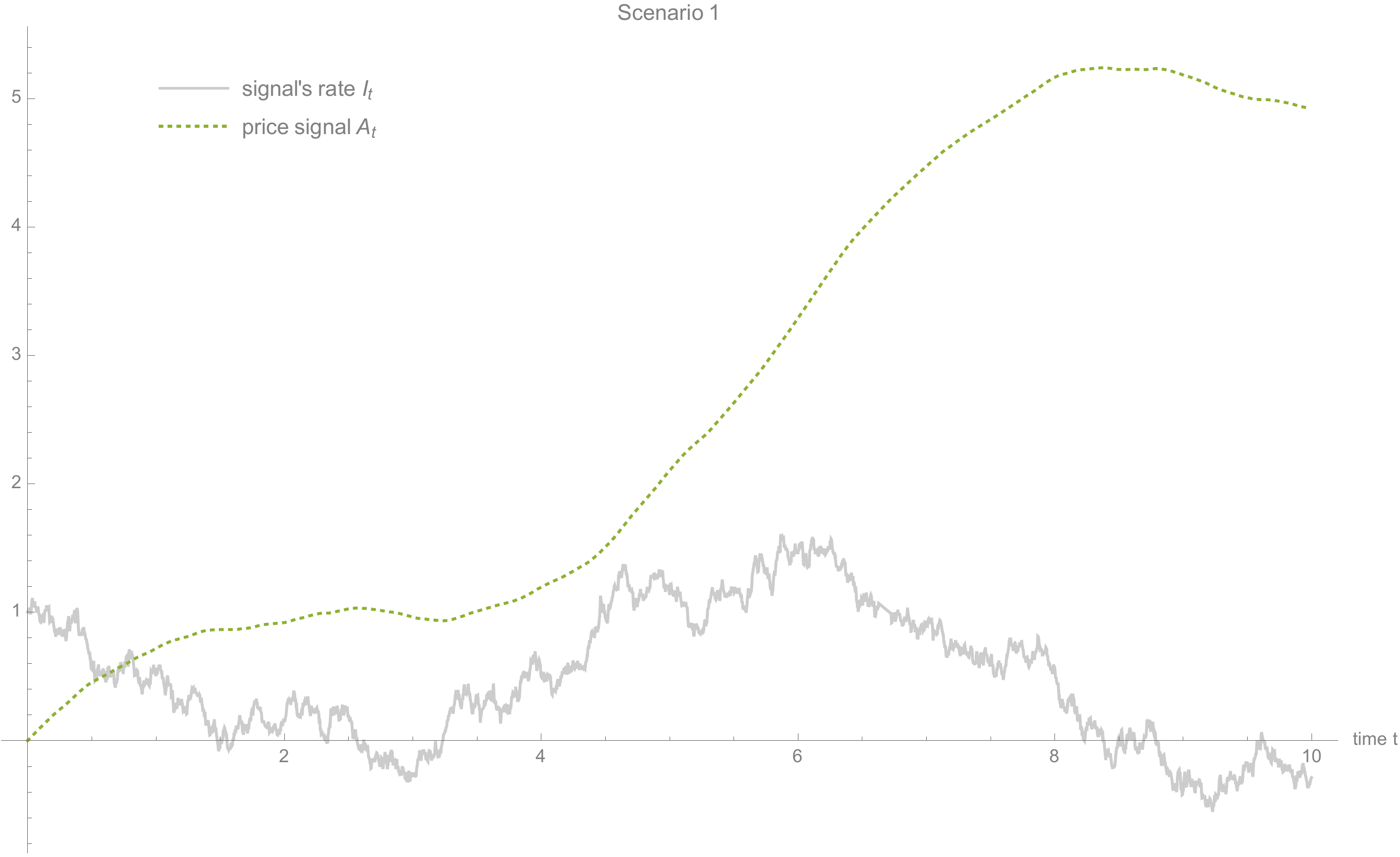}
    \par \vspace{2em}
    \includegraphics[scale=.5]{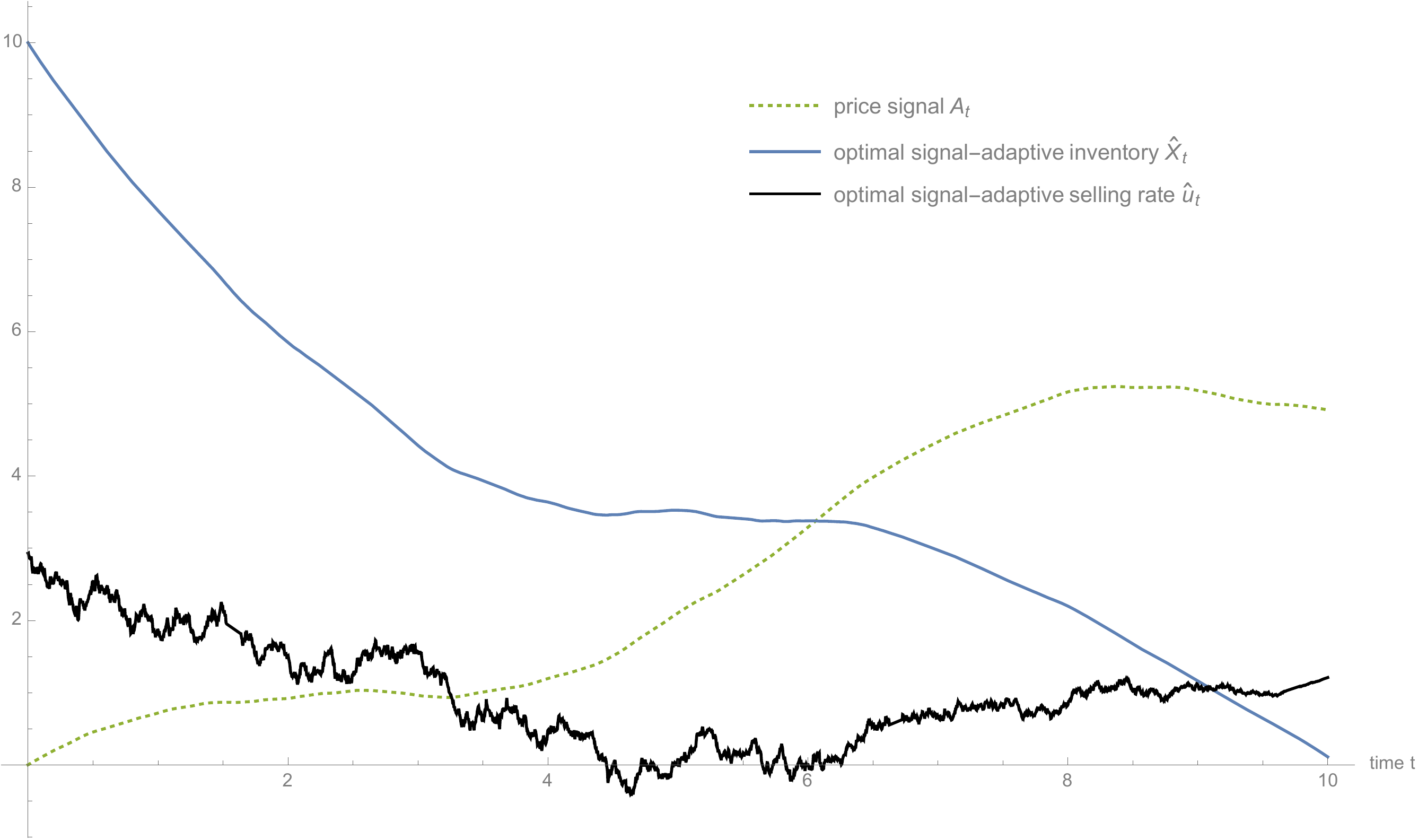}
    \caption{Upper panel: Realization of the signal rate $I_t$ (solid
      grey) and corresponding signal process $A_t$ (dashed
      green). Lower panel: Optimal signal-adaptive inventory (solid
      blue) and corresponding selling rate (solid black) for the same
      signal process (dashed green) as in the upper panel.}
    \label{fig:example1}
  \end{center}
\end{figure}

\begin{figure}[htbp!] 
  \begin{center}
    \includegraphics[scale=.5]{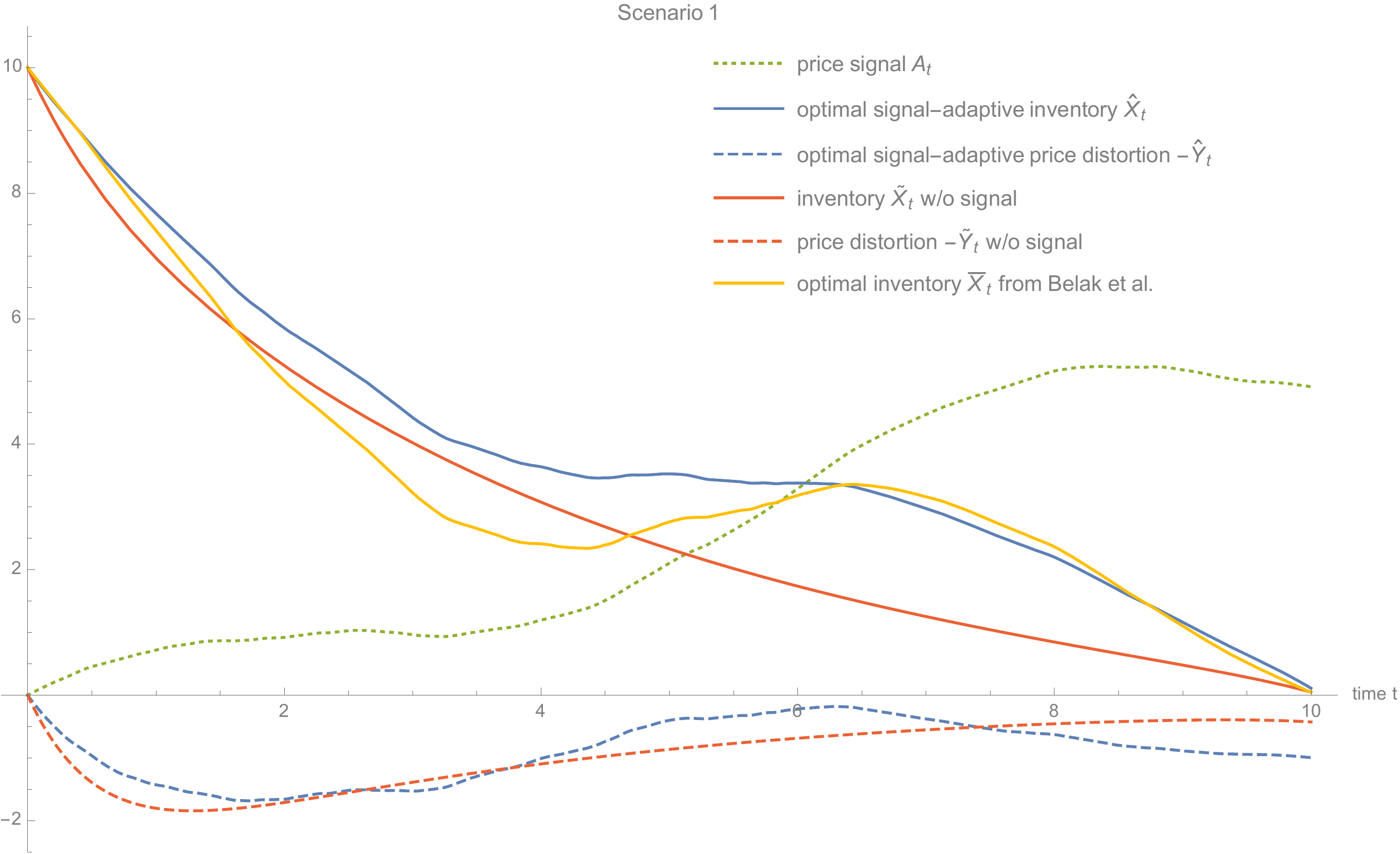}
    \caption{Comparison between the optimal signal-adaptive inventory
      (solid blue) and corresponding price distortion (dashed blue)
      with the optimal inventory (solid red) and corresponding price
      distortion (dashed red) ignoring the signal, as well as with the
      optimal signal-adaptive inventory with purely temporary price
      impact (solid yellow) for the same signal process (dashed green)
      as in Figure 1.}
    \label{fig:example12}
  \end{center}
\end{figure}

\begin{figure}[htbp!] 
  \begin{center}
    \includegraphics[scale=.50]{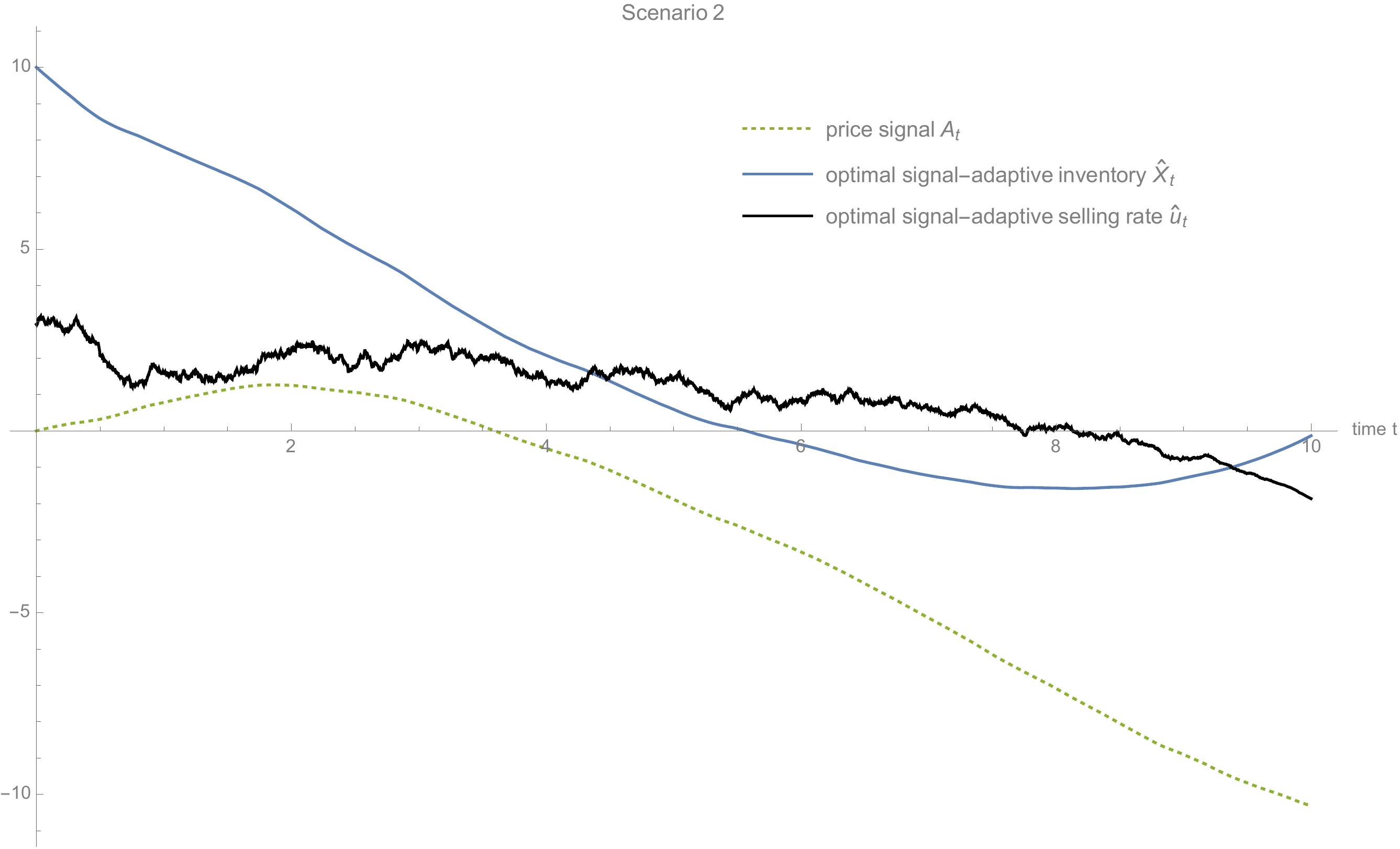}
    \par \vspace{2em}
    \includegraphics[scale=.50]{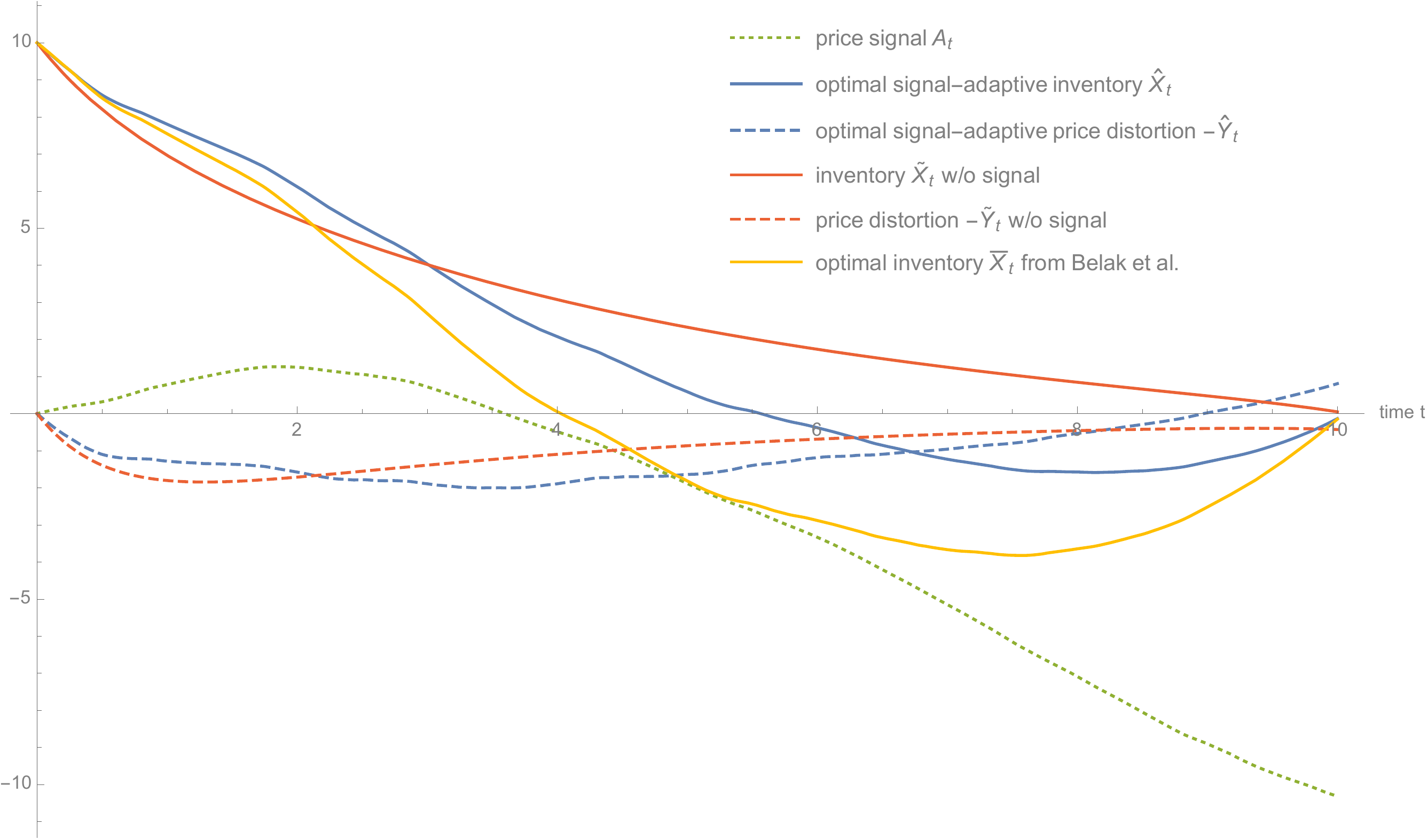}
    \caption{Upper panel: Similar to the lower panel of
      Figure~\ref{fig:example1}, optimal signal-adaptive inventory
      (solid blue) and corresponding selling rate (solid black) for a
      strongly decreasing price signal process (dashed green). Lower
      panel: Similar to Figure~\ref{fig:example12},
      comparison of the different optimal inventory and corresponding
      price distortion trajectories for the same signal process
      (dashed green) as in the upper panel.}
    \label{fig:example2}
  \end{center}
\end{figure}

\begin{figure}[htbp!] 
  \begin{center}
    \includegraphics[scale=.50]{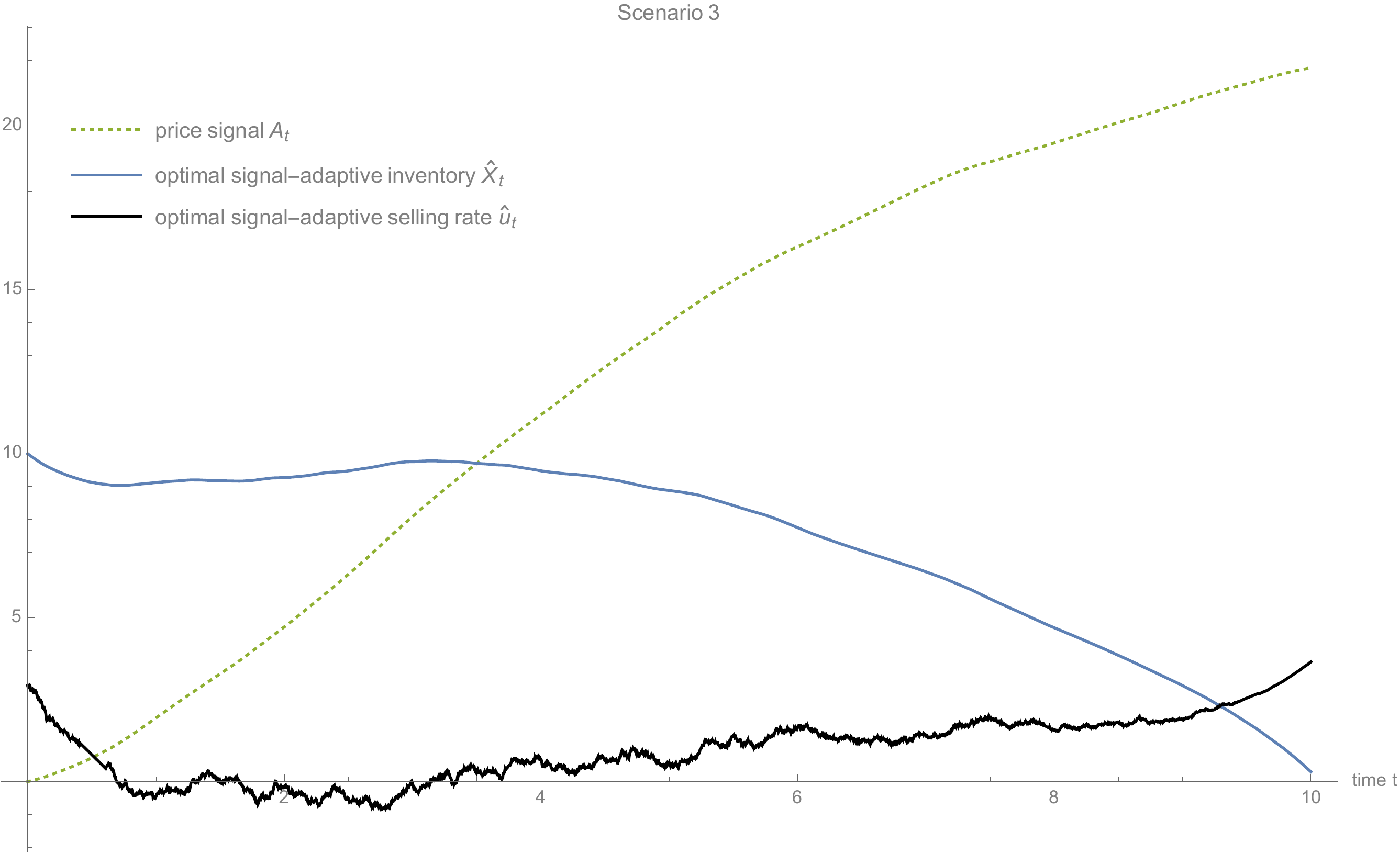}
    \par \vspace{2em}
    \includegraphics[scale=.50]{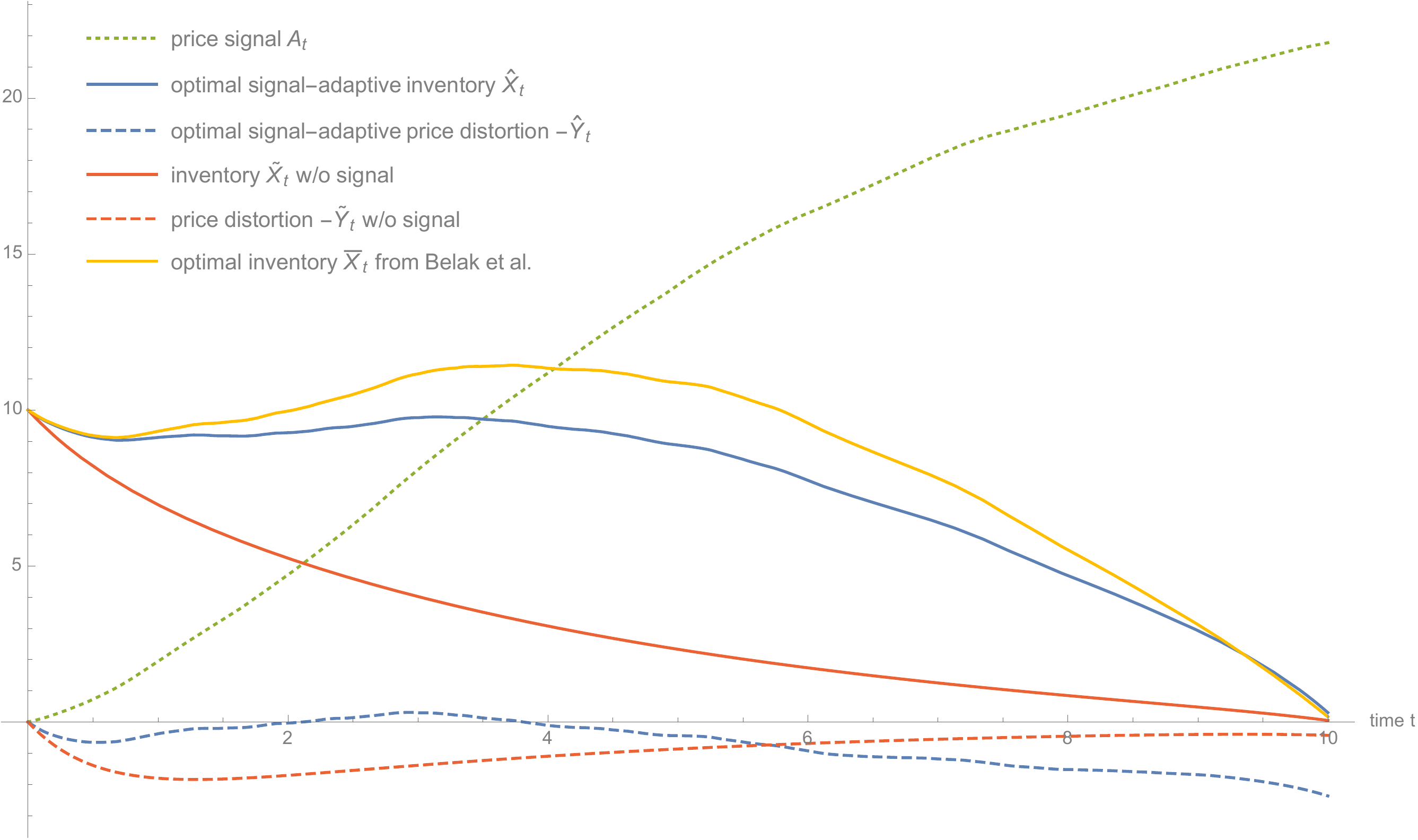}
    \caption{Upper panel: Similar to the upper panel of
      Figure~\ref{fig:example2}, optimal signal-adaptive inventory
      (solid blue) and corresponding selling rate (solid black) for a
      strongly increasing price signal process (dashed green). Lower
      Panel: Similar to the lower panel of Figure~\ref{fig:example2},
      comparison of the different optimal inventory and corresponding
      price distortion trajectories.}
    \label{fig:example3}
  \end{center}
\end{figure}

Figure~\ref{fig:example1} on the upper panel shows a realization of
the Ornstein-Uhlenbeck signal rate process $(I_t)_{0 \leq t \leq T}$
as in~\eqref{I-OU} in solid grey together with the resulting price
signal process $(A_t)_{0 \leq t \leq T}$ from~\eqref{a-i} in dashed
green. On the lower panel, we illustrate the corresponding optimal
signal-adaptive inventory $\hat{X}$ together with its selling rate
$\hat{u}$. In Figure~\ref{fig:example12} we compare the optimal
signal-adaptive inventory $\hat{X}$ and corresponding price distortion
$\hat{Y}$ (depicted in solid and dashed blue, respectively) for the
same price signal trajectory from Figure~\ref{fig:example1} with the
optimal inventory $\tilde X$ and price distortion~$\tilde Y$ ignoring
the signal (depicted in solid and dashed red, respectively). We also
plot in solid yellow the optimal inventory $\bar{X}$ for the purely
temporary price impact case. Interestingly, one can observe some
differences between the optimal strategies within the different
frameworks. As expected, in contrast to the strictly decreasing
inventory~$\tilde X$ ignoring the price signal process, the
signal-adaptive inventories $\hat{X}$ and~$\bar X$ utilize their
information about the upward trend of the latter and slow down the
liquidation of the risky asset midway. Moreover, the inventory
$\bar X$ taking into account only temporary price impact does so more
aggressively which results in trading also in the opposite direction
and buying some shares of the risky asset amid its liquidation
schedule.

Figures~\ref{fig:example2} and~\ref{fig:example3} illustrates in
similar fashion the optimal inventory, selling rate and price
distortion trajectories for two extreme scenarios: A strongly
decreasing price signal (Figure~\ref{fig:example2}) and a strongly
increasing price signal (Figure~\ref{fig:example3}). Again, we observe
that in a purely temporary price impact setup the trader tends to take
more risks by trading more boldly in the opposite direction to her
selling intentions in order to profit from the perceived information
about the price signal's tendencies. In fact, recall that the feedback
form of the optimal selling rate $\hat u$ in~\eqref{u-sol2}
compensates for the induced price distortion $Y^{\hat{u}}$. It is
therefore sensible to expect that this results in the observed
deceleration of the overall turnover rate as shown by the graphs in
Figures~\ref{fig:example1} to~\ref{fig:example3}.

\section{Proofs} \label{sec:proof}

\subsection{Proof of Theorem~\ref{thm:main}} \label{subsec:proofmainthm}

In fact, the probabilistic and convex analytic calculus of variations approach from~\citet{BankSonerVoss:17} can be brought to bear to 
prove our main Theorem~\ref{thm:main}. Indeed, note that for any $u \in \mathcal{A}$ the map $u \mapsto J(u)$ in~\eqref{def:objective} is strictly concave. Therefore, it admits a unique maximizer characterized by the critical point at which the G\^ateaux derivative
\begin{equation} \label{def:gateaux}
    \langle J'(u), \alpha \rangle \triangleq \lim_{\varepsilon \rightarrow 0} \frac{J(u + \varepsilon \alpha) - J(u)}{\varepsilon}
\end{equation}
of the functional $J$ vanishes for any direction $\alpha = (\alpha_t)_{0 \leq t \leq T} \in \mathcal{A}$; see, e.g.,~\cite{EkelTem:99}. The G\^ateaux derivative in~\eqref{def:gateaux} can be readily computed.

\begin{lemma}  \label{lem:gateaux}
For $u \in \mathcal{A}$ we have
\begin{equation}
\begin{aligned} \label{eq:gateaux}
   \langle J'(u), \alpha \rangle  &= \mathbb{E} \Bigg[\int_0^T \alpha_s \left( P_s - \kappa Y^u_s - \kappa \int_s^T e^{-\rho (t-s)} \gamma u_t dt - 2 \lambda u_s \right. \Bigg. \\
    & \hspace{75pt} \Bigg. \Bigg. + 2 \phi \int_s^T X^u_t dt + 2 \varrho X^u_T - P_T \Bigg) ds \Bigg]
\end{aligned}
\end{equation}
for any $\alpha \in \mathcal{A}$.
\end{lemma} 

\begin{proof} 
Let $\varepsilon > 0$ and $u, \alpha \in \mathcal{A}$. Note that $X_t^{u+\varepsilon\alpha} = X_t^u - \varepsilon \int_0^t \alpha_s ds$ and $Y_t^{u+\varepsilon\alpha} = Y^u_t + \varepsilon \gamma \int_0^t e^{-\rho (t-s)} \alpha_s ds$ for all $t \in [0,T]$. Next, since 
\begin{equation*}
\begin{aligned}
    & J(u+\varepsilon \alpha) - J(u) \\
    & = \varepsilon \; \mathbb{E} \Bigg[ \int_0^T \left( P_t - \kappa Y^u_t \right) \alpha_t dt - \kappa \int_0^T u_t \left( \int_0^t e^{-\rho (t-s)} \gamma \alpha_s ds \right) dt \\
    & \hspace{40pt} - 2 \lambda \int_0^T u_t \alpha_t dt + 2 \phi \int_0^T X^u_t \left( \int_0^t \alpha_s ds \right) dt \\
    & \hspace{40pt} \Bigg. +2 \varrho X_T^u \int_0^T \alpha_s ds - P_T  \int_0^T \alpha_s ds \Bigg] \\
    & \quad + \varepsilon^2 \; \E \Bigg[ -\kappa \gamma\int_0^T \left(   \int_0^t e^{-\rho (t-s)}  \alpha_s ds\right )  \alpha_t dt  - \lambda \int_0^T  \alpha_t^2 dt \Bigg. \\
    & \hspace{55pt}  - \phi \int_0^T \left( \int_0^t \alpha_s ds \right)^2 dt -\varrho \left(  \int_0^T \alpha_s ds \right)^{2} \Bigg]
\end{aligned}
\end{equation*}
we obtain the desired result in~\eqref{eq:gateaux} after applying Fubini's theorem twice. Also observe that all terms are finite since $u, \alpha \in \mathcal{A}$.
\end{proof}

Given the explicit expression of the G\^ateaux derivative in~\eqref{eq:gateaux} we can now derive a first order optimality condition. It takes the form of a coupled system of linear forward backward stochastic differential equations.

\begin{lemma} \label{lem:FOC}
A control $\hat{u} \in \mathcal A$ solves the optimization problem in~\eqref{def:optimization} if and only if the processes $(X^{\hat{u}},Y^{\hat{u}},{\hat{u}},Z^{\hat{u}})$ satisfy following coupled linear forward backward SDE system 
\begin{equation} \label{eq:FBSDE}
\left\{
\begin{aligned}
    dX^u_t = & \, - u_t dt, \quad X^u_0 = x\\
    dY^u_t = & \, -\rho Y_t^u dt + \gamma u_t dt, \quad Y^u_0 = y \\
    du_t = & \, \frac{dP_t}{2\lambda}  + \frac{\kappa\rho Y^u_t }{2\lambda} dt - \frac{\phi X^u_t}{\lambda} dt + \frac{\rho Z^u_t }{2\lambda} dt + dM_t, \quad u_T = \frac{\varrho X^u_T }{\lambda} - \frac{\kappa Y^u_T}{2\lambda}, \\
    dZ^u_t = & \,\rho Z^u_t dt + \kappa \gamma u_t dt + dN_t, \quad  Z^u_T = 0,
\end{aligned}
\right.
\end{equation}
for two suitable square integrable martingales $M=(M_t)_{0\leq t \leq T}$ and $N=(N_t)_{0\leq t\leq T}$.
\end{lemma} 

\begin{remark}
The appearance of the auxiliary process $Z^u$ in the above FBSDE system~\eqref{eq:FBSDE} is very natural in our setup because of the two-dimensional controlled state variable $(X^u,Y^u)$ in~\eqref{def:X} and~\eqref{def:Y}. In fact, the two processes $u$ and $Z^u$ satisfying the BSDEs in~\eqref{eq:FBSDE} correspond to the two associated so-called adjoint processes which are arising in Pontryagin's stochastic maximum principle; see, e.g.,~\citet{Carmona:16}, Chapter 4.2.
\end{remark}

\begin{proof} 
Since we are maximizing the strictly concave functional $u \mapsto J(u)$ over $\mathcal A$, a necessary and sufficient condition for the optimality of $\hat{u} \in \mathcal A$ with corresponding controlled state processes $X^{\hat{u}}$ and $Y^{\hat{u}}$ in \eqref{def:X} and \eqref{def:Y}, respectively, is given by
\begin{equation*}
     \langle J'(\hat{u}), \alpha \rangle = 0 \; \text{for all } \alpha \in \mathcal A;
\end{equation*}
cf., e.g.,~\cite{EkelTem:99}. By Lemma \ref{lem:gateaux} this condition is equivalent to
\be \label{proof:FOC:eqFOC}
\begin{aligned} 
&\mathbb{E} \Bigg[\int_0^T \alpha_s \bigg( P_s - \kappa Y^{\hat{u}}_s - \kappa \int_s^T e^{-\rho (t-s)} \gamma \hat{u}_t dt - 2 \lambda u_s \\
& \hspace{58pt} \Bigg. \Bigg. + 2 \phi \int_s^T X^{\hat{u}}_t dt   + 2 \varrho X^{\hat{u}}_T - P_T \bigg) ds \Bigg] =0
\end{aligned} 
\ee 
for all $\alpha \in \mathcal{A}$. In the following we will argue that $\hat{u} \in \mathcal A$ with $(X^{\hat{u}},Y^{\hat{u}})$ satisfies the first order condition in~\eqref{proof:FOC:eqFOC} if and only if $(X^{\hat{u}},Y^{\hat{u}},{\hat{u}},Z^{\hat{u}})$ satisfy the FBSDE system in~\eqref{eq:FBSDE}.

\emph{Necessity:} Although necessity follows from the uniqueness
  of the optimal solution together with the sufficiency argument below, we give
  the complete proof here in order to shed light on the derivation
  of the FBSDE system in \eqref{eq:FBSDE}. 

  Assume that $\hat u \in \mathcal A$ maximizes $J$, i.e., the first order condition in~\eqref{proof:FOC:eqFOC} is satisfied. Then, by applying optional projection we also obtain that
\begin{equation} \label{proof:FOC:eq1}
\begin{aligned}
& \mathbb{E} \left[ \int_0^T \alpha_s \left( P_s - \kappa Y^{\hat{u}}_s - \kappa \, \mathbb{E}_s \left[ \int_s^T e^{-\rho (t-s)} \gamma \hat{u}_t dt \right] - 2 \lambda \hat{u}_s \right. \right. \\
& \left. \left. \hspace{58pt} + \, \mathbb{E}_s \left[ 2 \phi \int_s^T X^{\hat{u}}_t dt + 2 \varrho X^{\hat{u}}_T - P_T \right] \right) ds \right] = 0
\end{aligned}
\end{equation}
for all $\alpha \in \mathcal A$. But this implies that
\be \label{proof:FOC:eq2}
\begin{aligned}
   & P_s - \kappa Y^{\hat{u}}_s - \kappa e^{\rho s} \mathbb{E}_s \left[ \int_s^T e^{-\rho t} \gamma \hat{u}_t dt \right] - 2 \lambda \hat{u}_s + \mathbb{E}_s \left[2 \phi \int_s^T X^{\hat{u}}_t dt + 2 \varrho X^{\hat{u}}_T - P_T \right]  \\
    & = \, P_s - \kappa Y^{\hat{u}}_s - \kappa e^{\rho s} \left( \mathbb{E}_s \left[ \int_0^T e^{-\rho t} \gamma \hat{u}_t dt \right] - \int_0^s e^{-\rho t} \gamma \hat{u}_t dt \right) - 2 \lambda \hat{u}_s   \\
    & \hspace{12pt} + \mathbb{E}_s \left[ 2 \phi \int_0^T X^{\hat{u}}_t dt + 2 \varrho X^{\hat{u}}_T - P_T \right] - 2 \phi \int_0^s X^{\hat{u}}_t dt \\
    & = 0 \quad d\P \otimes ds\textrm{-a.e. on } \Omega \times [0,T].
\end{aligned}
\ee
Now, by introducing the square integrable martingales 
\be \label{proof:FOC:martingales} 
\tilde{M}_s \triangleq \mathbb{E}_s \left[ 2 \phi \int_0^T X^{\hat{u}}_t dt + 2 \varrho X^{\hat{u}}_T - P_T \right], \quad
\tilde{N}_s \triangleq \mathbb{E}_s \left[ \int_0^T e^{-\rho t} \gamma \hat{u}_t dt \right],
\ee
as well as the auxiliary square integrable process 
\begin{equation} \label{proof:FOC:defZ} 
    Z^{\hat{u}}_s \triangleq \kappa e^{\rho s} \left( \int_0^s e^{-\rho t} \gamma \hat{u}_t dt - \tilde{N}_s \right)
\end{equation}
for all $s \in [0,T]$ (note that $P_T \in L^2(\Omega,\mathcal F_T, \P)$ because of~\eqref{ass:P}; and that $u \in L^2(\P \times [0,T])$ which also implies $X^{\hat{u}}_T \in L^2(\Omega,\mathcal F_T, \P)$), we can rewrite~\eqref{proof:FOC:eq2} as 
\be \label{proof:FOC:eq3}
\begin{aligned}
& P_s - \kappa Y^{\hat{u}}_s - \kappa e^{\rho s} \left(\tilde{N}_s - \int_0^s e^{-\rho t} \gamma \hat{u}_t dt \right) - 2 \lambda \hat{u}_s + \tilde{M}_s - 2 \phi \int_0^s X^{\hat{u}}_t dt \\
& = P_s - \kappa Y^{\hat{u}}_s + Z^{\hat{u}}_s - 2 \lambda \hat{u}_s + \tilde{M}_s - 2 \phi \int_0^s X^{\hat{u}}_t dt \\
& = 0 \quad d\P \otimes ds\textrm{-a.e. on } \Omega \times [0,T].
\end{aligned}
\ee
Note that $Z^{\hat{u}}$ in~\eqref{proof:FOC:defZ} satisfies the BSDE
\begin{equation} \label{proof:FOC:ZBSDE} 
    dZ^{\hat{u}}_t = \rho Z^{\hat{u}}_t dt + \kappa \gamma \hat{u}_t dt - \kappa e^{\rho t} d\tilde{N}_t, \qquad Z^{\hat{u}}_T = 0. 
\end{equation}
Also observe that the controlled forward dynamics of $Y^{\hat{u}}$ in~\eqref{def:Y} satisfy 
\begin{equation} \label{proof:FOC:YSDE} 
    Y^{\hat{u}}_0 = y, \quad dY^{\hat{u}}_t = - \rho Y^{\hat{u}}_t dt + \gamma \hat{u}_t dt.
\end{equation}
Hence, it follows from the representation in~\eqref{proof:FOC:eq3} that $\hat{u}$ satisfies the BSDE 
\be \label{proof:FOC:uBSDE}
\begin{aligned}
    d\hat{u}_s = & \, \frac{dP_s }{2\lambda} - \frac{\kappa}{2\lambda} dY^{\hat{u}}_s - \frac{\phi X^{\hat{u}}_s}{\lambda}  ds + \frac{\rho Z^{\hat{u}}_s}{2\lambda}  ds + \frac{\kappa\gamma}{2\lambda} \hat{u}_s ds + \frac{d\tilde{M}_s}{2\lambda}  - \frac{\kappa e^{\rho s}}{2\lambda}  d\tilde{N}_s\\
    = & \frac{dP_s }{2\lambda} + \frac{\kappa\rho Y^{\hat{u}}_s}{2\lambda} ds - \frac{\phi X^{\hat{u}}_s}{\lambda} ds + \frac{\rho Z^{\hat{u}}_s}{2\lambda}  ds + \frac{d\tilde{M}_s}{2\lambda}  - \frac{\kappa e^{\rho s}}{2\lambda}  d\tilde{N}_s, \\
    \hat{u}_T = & \, \frac{\rho X^{\hat{u}}_T}{\lambda}  - \frac{\kappa Y^{\hat{u}}_T}{2\lambda}.
\end{aligned}
\ee
Consequently, together with the forward dynamics of $X^{\hat{u}}$ in~\eqref{def:X}, we can conclude from \eqref{proof:FOC:YSDE}, \eqref{proof:FOC:uBSDE} and \eqref{proof:FOC:ZBSDE} that the processes $(X^{\hat{u}}, Y^{\hat{u}},\hat{u},Z^{\hat{u}})$ satisfy the FBSDE system in~\eqref{eq:FBSDE} with suitably chosen square integrable martingales $M=(M_t)_{0\leq t \leq T}$ and $N=(N_t)_{0\leq t\leq T}$ in terms of $\tilde{M}=(\tilde{M}_t)_{0\leq t \leq T}$ and $\tilde{N}=(\tilde{N}_t)_{0\leq t\leq T}$ given in~\eqref{proof:FOC:martingales}.

\emph{Sufficiency:} Let us now assume that $(X^{\hat{u}}, Y^{\hat{u}}, \hat{u}, Z^{\hat{u}})$ is a solution to the FBSDE system in~\eqref{eq:FBSDE} and $\hat{u} \in \mathcal A$. We have to show that $\hat{u}$ with controlled states $(X^{\hat{u}}, Y^{\hat{u}})$ satisfies the first order condition in~\eqref{proof:FOC:eqFOC} or, equivalently, in~\eqref{proof:FOC:eq1}. To this end, first note that the unique strong solution $\hat{u}$ to the associated  linear backward SDE in~\eqref{eq:FBSDE} is indeed given by~\eqref{proof:FOC:eq3}, i.e.,
\begin{equation*}
    2 \lambda \hat{u}_s = P_s - \kappa Y^{\hat{u}}_s - \kappa e^{\rho s} \left(\tilde{N}_s - \int_0^s e^{-\rho t} \gamma \hat{u}_t dt \right) + \tilde{M}_s - 2 \phi \int_0^s X^{\hat{u}}_t dt
\end{equation*}
with $\tilde{M}$ and $\tilde{N}$ as defined in~\eqref{proof:FOC:martingales}. Plugging this into~\eqref{proof:FOC:eqFOC} and applying Fubini's theorem yields
\bn
&& \mathbb{E} \left[ \int_0^T \alpha_s \left( \kappa e^{\rho s} \left( \tilde{N}_s - \int_0^T e^{-\rho t} \gamma \hat{u}_t dt \right) - \tilde{M}_{s}  + 2 \phi \int_0^T X^{\hat{u}}_t dt + 2 \varrho X^{\hat{u}}_T - P_{T} \right) dt \right] \\
&& = \mathbb{E} \left[ \int_0^T \alpha_s \left(\kappa e^{\rho s} \left( \tilde{N}_s  - \tilde{N}_{T} \right) - \tilde{M}_{s} + \tilde{M}_{T} \right) dt \right] \\
&& = \mathbb{E} \left[ \int_0^T \alpha_s \left( \kappa e^{\rho s} \left( \tilde{N}_s - \E_s[ \tilde{N}_{T}] \right) + \E_s[\tilde{M}_{T}] - \tilde{M}_{s} \right) dt \right] = 0
\en
for all $\alpha \in \mathcal A$ since $\tilde{N}$ and $\tilde{M}$ are martingales. Consequently, the first order condition in \eqref{proof:FOC:eqFOC} is satisfied and $\hat{u} \in \mathcal A$ is optimal.
\end{proof} 
 
\noindent\textbf{Proof of Theorem~\ref{thm:main}.} \emph{Step 1:} In view of Lemma~\ref{lem:FOC} we have to solve the linear FBSDE system in~\eqref{eq:FBSDE}. One possibility to achieve this is to adapt the approach in~\cite{BMO:19}. Introducing
\begin{equation} \label{def:XM}
    \boldsymbol{X}^u_t \triangleq \begin{pmatrix}
    X^u_t \\ Y^u_t \\ u_t \\ Z^u_t
    \end{pmatrix},   
    \, \quad  \boldsymbol{M}_t \triangleq \begin{pmatrix}
    0 \\ 0 \\ P_t - 2 \lambda M_t \\ 2 \lambda N_t
    \end{pmatrix} \quad (0 \leq t \leq T),
\end{equation}
the linear system in~\eqref{eq:FBSDE}, together with matrix $L$ in~\eqref{def:A}, can be written as
\begin{equation} \label{eq:linODE}
    d\boldsymbol{X}^u_t = L \boldsymbol{X}^u_t dt + \frac{1}{2\lambda} d\boldsymbol{M}_t \quad (0 \leq t \leq T)
\end{equation}
with initial conditions $\boldsymbol{X}^{u,1}_0 = x$, $\boldsymbol{X}^{u,2}_0 = y$ and terminal conditions
\begin{equation} \label{def:termCond}
 \left( \varrho/\lambda, -\kappa/(2\lambda),-1,0 \right) \boldsymbol{X}^u_T = 0 \quad \text{and} \quad (0,0,0,1) \boldsymbol{X}^u_T = 0.
\end{equation}
Observe that the unique solution of~\eqref{eq:linODE} can be represented in terms of the matrix exponential $S(t) = \exp(A t) = (S_{ij}(t))_{1 \leq i,j, \leq 4}$ introduced in~\eqref{def:S} as  
\begin{equation} \label{eq:X-T}
    \boldsymbol{X}^u_T = S(T-t) \boldsymbol{X}^u_t + \frac{1}{2\lambda} \int_t^T S(T-s) d\boldsymbol{M}_s \quad (0 \leq t \leq T).
\end{equation}
Next, following the same idea as in the proof of Theorem 3.1 in~\cite{BMO:19}, we use the first terminal condition in~\eqref{def:termCond} and multiply~\eqref{eq:X-T} by $( \varrho/\lambda, -\kappa/(2\lambda),-1,0)$ to obtain
\begin{align*} 
0 = & \; G(T-t) \boldsymbol{X}^u_t + \frac{1}{2\lambda} \int_t^T G(T-s) d\boldsymbol{M}_s \nonumber \\
= & \; G_1(T-t) X^u_t + G_2(T-t) Y^u_t + G_3(T-t) u_t + G_4(T-t) Z^u_t \nonumber \\
& + \frac{1}{2\lambda} \int_t^T G_3(T-s) (dP_s - 2\lambda dM_s) + \int_t^T G_4(T-s)dN_s \qquad (0 \leq t \leq T)
\end{align*}
with $G=(G_i)_{i=1,\ldots,4}$ as defined in~\eqref{def:G}. Since $G_3(T-t) \neq 0$ for all $t \in [0,T]$ (see Lemma~\ref{lemma-g-s} (iii) below) we get
\begin{equation}
\begin{aligned} 
u_t = & -\frac{G_1(T-t)}{G_3(T-t)} X^u_t - \frac{G_2(T-t)}{G_3(T-t)} Y^u_t - \frac{G_4(T-t)}{G_3(T-t)} Z^u_t \nonumber \\
& - \frac{1}{2\lambda} \int_t^T \frac{G_3(T-s)}{G_3(T-t)} (dP_s - 2\lambda dM_s) - \int_t^T \frac{G_4(T-s)}{G_3(T-t)} dN_s \qquad (0 \leq t \leq T). 
\end{aligned}
\end{equation}
Taking conditional expectation in the latter equation and using the fact that $P \in \mathcal H^2$, as well as that $M$, $N$ are square integrable martingales, we arrive at the identity 
\begin{equation} \label{eq:u_Rep}
\begin{aligned} 
u_{t} = & -\frac{G_1(T-t)}{G_3(T-t)}X^u_{t}  -\frac{G_2(T-t)}{G_{3}(T-t)}Y^u_{t}  - \frac{G_4(T-t)}{G_3(T-t)}Z^u_t \\
 & - \frac{1}{2\lam}\mathbb E_{t} \left[ \int_{t}^T\frac{G_3(T-s)}{G_{3}(T-t)}dA_{s} \right] \qquad (0 \leq t \leq T).
\end{aligned}
\end{equation}
Moreover, using also the second terminal condition in~\eqref{def:termCond} and multiplying~\eqref{eq:X-T} by $(0, 0, 0, 1)$ gives us 
\begin{align*} 
0 = & \; S_{4,\cdot}(T-t) \boldsymbol{X}^u_t + \frac{1}{2\lambda} \int_t^T S_{4,\cdot}(T-s) d\boldsymbol{M}_s \nonumber \\
= & \; S_{4,1}(T-t) X^u_t + S_{4,2}(T-t) Y^u_t + S_{4,3}(T-t) u_t + S_{4,4}(T-t) Z^u_t \nonumber \\
& + \frac{1}{2\lambda} \int_t^T S_{4,3}(T-s) (dP_s - 2\lambda dM_s) + \int_t^T S_{4,4}(T-s)dN_s \qquad (0 \leq t \leq T),
\end{align*}
where we used the notation $S_{4,\cdot} = (S_{4,1}, S_{4,2}, S_{4,3}, S_{4,4})$. Since $S_{4,4}(T-t) \neq 0$ for all $t \in [0,T]$ (see Lemma~\ref{lemma-g-s} (ii) below), solving for $Z^u$ and taking once more conditional expectation as above yields
\begin{equation} \label{eq:Z_Rep}
\begin{aligned} 
Z^u_t = & - \frac{S_{4,1}(T-t)}{S_{4,4}(T-t)} X^u_t  - \frac{S_{4,2}(T-t)}{S_{4,4}(T-t)} Y^u_t - \frac{S_{4,3}(T-t)}{S_{4,4}(T-t)} u_t  \\
& - \frac{1}{2\lambda} \mathbb E_t \left[ \int_t^T \frac{S_{4,3}(T-s)}{S_{4,4}(T-t)} dA_s \right] \qquad (0 \leq t \leq T).
\end{aligned}
\end{equation}
Finally by (\ref{g-s-cond}), $v_0$ in~\eqref{def:vs} is well-defined. Hence, plugging~\eqref{eq:Z_Rep} into~\eqref{eq:u_Rep} and solving for $u$ yields
\begin{align*}
u_{t} = & \; v_0(T-t) \left(
          \frac{G_4(T-t)}{G_3(T-t)}\frac{S_{4,1}(T-t)}{S_{4,4}(T-t)}
          -\frac{G_1(T-t)}{G_3(T-t)} \right) X^u_{t}  \\
& + v_0(T-t) \left(
   \frac{G_4(T-t)}{G_3(T-t)}\frac{S_{4,2}(T-t)}{S_{4,4}(T-t)}
   -\frac{G_2(T-t)}{G_3(T-t)} \right) Y^u_{t} \\
& +\frac{1}{2\lambda} v_0(T-t) \left( \frac{G_4(T-t)}{G_3(T-t)}\E_t \left[
   \int_t^T  \frac{S_{4,3}(T-s)}{S_{4,4}(T-t)}dA_s \right] - \E_{t} \left[
   \int_{t}^T\frac{G_3(T-s)}{G_{3}(T-t)}dA_{s} \right] \right) 
\end{align*}
as claimed in~\eqref{eq:opt_u}. 

\emph{Step 2:} It remains to argue that $\hat{u}$ in~\eqref{eq:opt_u} belongs to $\mathcal A$. First, due to Lemma~\ref{lemma-g-s}~(i) and assumption~\eqref{g-s-cond} we can conclude that 
\bd 
\sup_{0\leq t\leq T} |v_0(T-t)|< \infty.
\ed
Together with Lemma~\ref{lemma-g-s}~(ii) and (iii) it then follows that the coefficients in front of $X^{\hat{u}}$ and $Y^{\hat{u}}$ in~\eqref{eq:opt_u} are bounded on $[0,T]$. By the same arguments we obtain that there exists a constant $C > 0$ such that
\bn
&&\sup_{0\leq t\leq T}  \left\{ \frac{G_4(T-t)}{G_3(T-t)}\E_t \left[
   \int_t^T  \frac{S_{4,3}(T-s)}{S_{4,4}(T-t)}dA_s \right] - \E_{t} \left[
   \int_{t}^T\frac{G_3(T-s)}{G_{3}(T-t)}dA_{s} \right] \right\} \\
   &&\leq  C \, \mathbb E \left[ \int_0^T|dA_t| \right] < \infty
\en
due to~\eqref{ass:P}. Together with~\eqref{def:X} and~\eqref{def:Y} we use these bounds in~\eqref{eq:opt_u} to get 
$$
\mathbb{E}[\hat{u}_t^2] \leq C_1+ C_2\int_0^t \mathbb{E}[\hat{u}_s^2] \, ds \quad (0 \leq t \leq T)   
$$
for some positive constants $C_1,C_2$. From Gronwall's lemma it then follows that 
$$
\sup_{0\leq t\leq T} \mathbb E[\hat{u}_t^2] < \infty, 
$$
so clearly $\hat{u} \in \mathcal A$ by Fubini's theorem.
\qed 

\medskip

\noindent\textbf{Proof of Corollary~\ref{cor:main}.} Simply observe that the optimally controlled state variable $\mathbb X^{\hat{u}} = (X^{\hat{u}},Y^{\hat{u}})$ prescribed in~\eqref{eq:opt_u} and~\eqref{def:Y} satisfies following two-dimensional linear (random) ordinary differential equation
\begin{equation} \label{eq:linODEsysXY}
    \mathbb X^{\hat{u}}_0 =(x,y),\quad d\mathbb X^{\hat{u}}_t = B(t) \, \mathbb X^{\hat{u}}_t dt + \hat{\zeta}_t \, b \, dt \quad (0 \leq t \leq T)
\end{equation}
with $(\hat{\zeta}_{t})_{0\leq t \leq T}$ as given in~\eqref{def:zeta} and $b = (-1, \gamma)^\top \in \mathbb R^2$. Hence, it follows from standard existence and uniqueness results for linear ODEs (cf., e.g., \citet{KaratzasShreve}, Section 5.6, and the references therein) that $\mathbb X^{\hat{u}}$ is given by
\begin{equation*}
\mathbb{X}^{\hat{u}}_t 
= \Phi(t) \left(
\mathbb{X}^{\hat{u}}_0 
+ \int_0^t \hat{\zeta}_s \, \Phi^{-1}(s) \, b \, ds 
\right) \quad (0 \leq t \leq T),
\end{equation*}
where $\Phi$ denotes the unique nonsingular solution to the matrix differential equation in~\eqref{eq:ODEPhi}.
\qed

\subsection{Computing the matrix exponential} \label{subsec:matrixexponential}

To compute the matrix exponential $S(t)= (S_{ij}(t))_{0 \leq i,j \leq 4} = e^{Lt} $ for all $t \in [0, \infty)$ in \eqref{def:matrixExp} we diagonalize matrix $L$
in~\eqref{def:A}, i.e., decompose $L = U D U^{-1}$ with diagonal
matrix $D \in \mathbb R^{4\times 4}$ and invertible matrix $U \in \mathbb R^{4\times 4}$. Then, it follows that
\begin{equation} \label{def:decomS}
S(t) = U e^{D t} U^{-1} \quad (t \geq 0),
\end{equation} 
where $e^{D t} \in \mathbb R^{4\times 4}$ is again a diagonal matrix. Introducing the constants $\theta \triangleq \lambda\rho^2 + \rho\kappa\gamma + \phi$, as well as 
\begin{equation} \label{def:c1c2}
c_1 \triangleq \frac{\theta}{\lambda} > 0, \quad
c_2 \triangleq \sqrt{\frac{(\theta - 2 \phi)^2 + 4
    \phi\rho\kappa\gamma}{\lambda^2}} > 0,
\end{equation}
it can be easily checked that the eigenvalues of $A$ are given by
\begin{equation} \label{def:eigenvalues}
\nu_1 \triangleq -\sqrt{\frac{c_1 - c_2}{2}}, \quad
\nu_2 \triangleq -\nu_1, \quad \nu_3 \triangleq
-\sqrt{\frac{c_1 + c_2}{2}}, \quad \nu_4 \triangleq -\nu_3,
\end{equation}
with corresponding eigenvectors
\begin{equation*} \label{def:eigenvector}
  v_i \triangleq
  \begin{pmatrix}
    \frac{\rho - \nu_i}{\kappa\gamma\nu_i} \\
     \frac{\nu_i - \rho}{\kappa(\nu_i+\rho)} \\
     \frac{\nu_i - \rho}{\kappa\gamma} \\
     1
    \end{pmatrix} \qquad (i=1,2,3,4).
\end{equation*}
Also note that $c_1 - c_2 > 0$ in~\eqref{def:eigenvalues}: Indeed, $c_1>c_2$ is equivalent to $\theta^2 > (\theta - 2\phi)^2 + 4 \phi\rho\kappa\gamma$ and hence to $-4\phi\lambda\rho^{2}<0$, which is satisfied. Consequently, we obtain that
\begin{equation} \label{def:DU}
  D \triangleq
  \begin{pmatrix}
    \nu_1 & 0 & 0 & 0 \\
    0 & -\nu_1 & 0 & 0 \\
    0 & 0 & \nu_3 & 0 \\
    0 & 0 & 0 & -\nu_3 \\
  \end{pmatrix}, \quad
    U \triangleq
  \begin{pmatrix}
    -\frac{\nu_1 - \rho}{\kappa\gamma\nu_1} & -\frac{\nu_1 + \rho}{\kappa\gamma\nu_1} & - \frac{\nu_3 - \rho}{\kappa\gamma\nu_3} & - \frac{\nu_3 + \rho}{\kappa\gamma\nu_3}  \\
     \frac{\nu_1 - \rho}{\kappa(\nu_1+\rho)} &
     \frac{\nu_1 + \rho}{\kappa(\nu_1-\rho)} &
     \frac{\nu_3 - \rho}{\kappa(\nu_3+\rho)} &
     \frac{\nu_3 + \rho}{\kappa(\nu_3-\rho)}  \\
     \frac{\nu_1 - \rho}{\kappa\gamma} & - \frac{\nu_1 +
     \rho}{\kappa\gamma} & \frac{\nu_3 - \rho}{\kappa\gamma}
     & - \frac{\nu_3 + \rho}{\kappa\gamma} \\
     1 & 1 & 1 & 1
  \end{pmatrix}
\end{equation}
satisfy $L = U D U^{-1}$ with
  \begin{align*} 
  & U^{-1} = \\
  & \; \frac{1}{4\rho^2 (\nu_1^2 -
             \nu_3^2)} \label{def:Uinv} \\
  &
  \begin{pmatrix}
    -2 \gamma \kappa \nu_1 \nu_3^2 
    (\nu_1 + \rho) & - \kappa (\nu_1^2 - \rho^2)
    (\nu_3^2 - \rho^2) & 2 \gamma \kappa \rho^2 (\nu_1 + \rho)  & - (\nu_3^2 - \rho^2)
    (\nu_1 + \rho)^2\\ 
     -2 \gamma \kappa \nu_1 \nu_3^2 
      (\nu_1 - \rho)&
     - \kappa (\nu_1^2 - \rho^2)
    (\nu_3^2 - \rho^2) &
     - 2 \gamma \kappa \rho^2 (\nu_1 - \rho) &
    -(\nu_3^2 - \rho^2)
    (\nu_1 - \rho)^2 \\
      2 \gamma \kappa \nu_1^2 \nu_3 
    (\nu_3 + \rho) & \kappa (\nu_1^2 - \rho^2)
    (\nu_3^2 - \rho^2) & -2 \gamma \kappa \rho^2 (\nu_3 + \rho) 
     & (\nu_1^2 - \rho^2)
    (\nu_3 + \rho)^2 \\
     2 \gamma \kappa \nu_1^2 \nu_3 
    (\nu_3- \rho) &  \kappa (\nu_1^2 - \rho^2)
    (\nu_3^2 - \rho^2) & 2 \gamma \kappa \rho^2 (\nu_3 -
    \rho) 
     & (\nu_1^2 - \rho^2)
    (\nu_3 - \rho)^2
  \end{pmatrix}.
\end{align*}
Thus, the matrix exponential $S(t) = (S_{ij}(t))_{1 \leq i,j \leq 4}$ in~\eqref{def:matrixExp} is given by
\begin{equation}
  S(t) = U \begin{pmatrix}
    e^{\nu_1 t} & 0 & 0 & 0 \\
    0 & e^{-\nu_1 t} & 0 & 0 \\
    0 & 0 & e^{\nu_3 t} & 0 \\
    0 & 0 & 0 & e^{-\nu_3 t}\\
  \end{pmatrix} U^{-1} \qquad (t \geq 0).
\end{equation}
In particular, due to the fact that all entries in the last row of $U$ in~\eqref{def:DU} are equal to one, we easily get that the functions $(S_{4,j}(T-t))_{1\leq j \leq 4}$ are given by
\begin{align} 
    S_{4,1}(T-t) = & \,
    \frac{\gamma\kappa\nu_1\nu_3}{\rho^2 (\nu_1^2 -
             \nu_3^2)} 
             \Big(  \nu_1\nu_3 \cosh( \nu_3(T-t)) - \nu_1\nu_3 \cosh(\nu_1(T-t)) \Big. \nonumber \\
    & \hspace{70pt} \Big. + \rho \nu_1 \sinh( \nu_3(T-t)) - \rho \nu_3 \sinh(\nu_1(T-t))\Big), \label{eq:S41} \\
    S_{4,2}(T-t) = & \,
    \frac{\kappa(\nu_1^2-\rho^2)(\nu_3^2-\rho^2)}{2\rho^2 (\nu_1^2 -
             \nu_3^2)} 
             \Big( \cosh( \nu_3(T-t)) - \cosh(\nu_1(T-t)) \Big),  \label{eq:S42} \\
    S_{4,3}(T-t) = & \,
    \frac{\gamma\kappa}{\nu_1^2 -
             \nu_3^2} 
             \Big(  \nu_1 \sinh( \nu_1(T-t)) - \nu_3 \sinh(\nu_3(T-t)) \Big. \nonumber \\
    & \hspace{45pt} \Big. + \rho \cosh( \nu_1(T-t)) - \rho \cosh(\nu_3(T-t))\Big), \label{eq:S43} \\
    S_{4,4}(T-t) = & \,
    \frac{1}{2\rho^2(\nu_1^2 -
             \nu_3^2)} 
             \Big( \big(\nu_1^2-\rho^2\big) \big(\nu_3^2+\rho^2\big) \cosh(\nu_3(T-t)) \Big. \nonumber \\
             & \hspace{76pt} - \big(\nu_1^2 + \rho^2\big) \big(\nu_3^2-\rho^2\big) \cosh(\nu_1(T-t)) \nonumber \\
    & \hspace{76pt} \Big. + 2 \rho \nu_3 \big( \nu_1^2 - \rho^2 \big) \sinh( \nu_3(T-t)) \nonumber
    \\ & \hspace{76pt} \Big. - 2\rho\nu_1 \big( \nu_3^2 - \rho^2 \big) \sinh( \nu_1(T-t))  \Big) \label{eq:S44}
\end{align}
for all $t \in [0,T]$. In addition, slightly more involved but still elementary computations reveal that the functions $G(t) = (G_i(t))_{1\leq 1 \leq 4} = (\varrho/\lambda, -\kappa/(2\lambda),-1,0) S(t)$ introduced in~\eqref{def:G} are given by
\begin{align} 
    G_{1}(T-t) & = \frac{1}{2\lambda \rho^2 (\nu_1^2 -
             \nu_3^2)} \label{eq:G1} \\
    & \hspace{17pt} \Big( \big( 2 \varrho \nu_3^2 (\nu_1^2 - \rho^2) + \gamma \kappa \nu_1^2 \nu_3^2 \big) \cosh( \nu_1(T-t)) \Big. \nonumber \\
    & \hspace{22pt} + \big( 2\lambda \nu_1\nu_3^2 (\nu_1^2 - \rho^2)  - \gamma\kappa\rho \nu_1\nu_3^2 \big) \sinh(\nu_1(T-t)) \nonumber \\
    & \hspace{22pt} - \big( 2 \varrho \nu_1^2 (\nu_3^2 - \rho^2) + \gamma \kappa \nu_1^2 \nu_3^2 \big) \cosh( \nu_3(T-t)) \nonumber \\
    & \hspace{22pt} \Big. - \big( 2\lambda \nu_1^2\nu_3 (\nu_3^2 - \rho^2) - \gamma\kappa\rho \nu_1^2\nu_3 \big) \sinh(\nu_3(T-t))
    \Big), \nonumber \\
    G_{2}(T-t) & = \frac{1}{4 \gamma \lambda \rho^2 \nu_1\nu_3 (\nu_1^2 -
             \nu_3^2)} \label{eq:G2} \\
    & \hspace{12pt} \Big( \nu_1\nu_3 (\nu^2_3 - \rho^2) \big( 2 (\nu_1^2 - \rho^2) (\varrho - \lambda\rho) + \gamma\kappa (\nu_1^2 + \rho^2) \big) \cosh( \nu_1(T-t)) \Big. \nonumber \\
    & \hspace{16pt} - \nu_3 (\nu^2_3 - \rho^2) \big( 2  (\nu_1^2 - \rho^2) (\varrho\rho - \lambda \nu_1^2) + 2 \gamma\kappa\rho\nu_1^2 \big) \sinh(\nu_1(T-t)) \nonumber \\
    & \hspace{16pt} - \nu_1\nu_3 (\nu^2_1 - \rho^2) \big( 2 (\nu_3^2 - \rho^2) (\varrho - \lambda\rho) + \gamma\kappa (\nu_3^2 + \rho^2) \big) \cosh(\nu_3(T-t)) \nonumber \\
    & \hspace{16pt} \Big. + \nu_1 (\nu^2_1 - \rho^2) \big( 2  (\nu_3^2 - \rho^2) (\varrho\rho - \lambda \nu_3^2) + 2 \gamma\kappa\rho\nu_3^2 \big) \sinh(\nu_3(T-t))
    \Big),  \nonumber\\
    G_{3}(T-t) & = \frac{1}{2\lambda\nu_1\nu_3 (\nu_1^2 -
             \nu_3^2)} \label{eq:G3} \\
    & \hspace{12pt} \Big( \nu_1\nu_3 \big( \gamma\kappa\rho  - 2\lambda (\nu_1^2 - \rho^2)\big) \cosh( \nu_1(T-t)) \Big. \nonumber \\
    & \hspace{16pt} - \nu_3 \big(\gamma\kappa\nu_1^2 +2\varrho (\nu_1^2-\rho^2) \big) \sinh(\nu_1(T-t)) \nonumber \\
    & \hspace{16pt} - \nu_1\nu_3 \big( \gamma\kappa\rho - 2 \lambda (\nu_3^2-\rho^2) \big) \cosh( \nu_3(T-t)) \nonumber \\
    & \hspace{16pt} \Big. + \nu_1 \big( \gamma\kappa\nu_3^2 + 2\varrho (\nu_3^2-\rho^2)\big) \sinh(\nu_3(T-t))
    \Big), \nonumber \\
    G_{4}(T-t) & = \frac{(\nu_1^2-\rho^2)(\nu_3^2-\rho^2)}{4\gamma\kappa \lambda\rho^2\nu_1\nu_3 (\nu_1^2 -
             \nu_3^2)} \label{eq:G4} \\
    & \hspace{12pt} \Big( \nu_1\nu_3  \big( 2 \varrho + \gamma\kappa + 2 \lambda \rho  \big) \cosh(\nu_1(T-t)) \Big. \nonumber \\
    & \hspace{16pt} + 2\nu_3 \big( \rho\varrho + \lambda \nu_1^2 \big) \sinh(\nu_1(T-t)) \nonumber \\
    & \hspace{16pt} - \nu_1\nu_3 \big( 2\varrho  + \gamma\kappa + 2\lambda\rho \big) \cosh(\nu_3(T-t)) \nonumber \nonumber \\
    & \hspace{16pt} \Big. - 2\nu_1 \big( \rho\varrho + \lambda \nu_3^2 \big) \sinh(\nu_3(T-t))
    \Big) \nonumber
\end{align}
for all $t \in [0,T]$. Finally, let us collect some useful properties of the eigenvalues $\nu_1$ and $\nu_3$ and the functions $S_{4,j}(\cdot), G_{j}(\cdot)$ for all $j\in \{1,\ldots,4\}$.

\begin{lemma} \label{lem:aux1}
For any positive constants $\lambda,\gamma, \kappa, \rho, \phi$ we have 
\be \label{nu-cond} 
\nu_1^2 \leq \rho^2  \leq \nu_3^2
\ee
with $\nu_1$ and $\nu_3$ given in~\eqref{def:eigenvalues}.
\end{lemma} 

\begin{proof}
First, from (\ref{def:c1c2}) and (\ref{def:eigenvalues}) we get 
\bn
\nu_1^2 - \rho^2 = \frac{\phi + \gamma \kappa \rho - \lambda \rho^2 - \sqrt{
 4 \gamma \kappa \phi \rho + (\phi - 
    \rho (\gamma \kappa + \lambda \rho))^2}}{2 \lambda}
\en
and we argue that
\be \label{ineq1}
\nu_1^2 - \rho^2 < 0
\ee
by considering following two cases: if $\phi + \gamma \kappa \rho - \lambda \rho^2 \leq 0$, then~\eqref{ineq1} holds trivially. 
Otherwise, if $\phi + \gamma \kappa \rho - \lambda \rho^2 > 0$, then~\eqref{ineq1} is equivalent to $-4\lambda\rho^3\kappa\gamma < 0$, which is satisfied. Next, again due to~\eqref{def:c1c2} and~\eqref{def:eigenvalues} we have 
$$
\nu_3^2 - \rho^2 = \frac{\phi + \gamma \kappa \rho - \lambda \rho^2 + \sqrt{
 4 \gamma \kappa \phi \rho + (\phi - 
    \rho (\gamma \kappa + \lambda \rho))^2}}{2 \lambda},
$$
and we obtain similarly that
\be \label{ineq2}
\nu_3^2 - \rho^2>0. 
\ee
Indeed, if $\phi + \gamma \kappa \rho - \lambda \rho^2 \geq0$, then (\ref{ineq2}) holds trivially. Otherwise, if $\phi + \gamma \kappa \rho - \lambda \rho^2 <0$, then (\ref{ineq2}) is again equivalent to $-4\lambda\rho^3\kappa\gamma < 0$, which is satisfied. Finally,~\eqref{ineq1} and~\eqref{ineq2} imply~\eqref{nu-cond}.
\end{proof}

\begin{lemma} \label{lemma-g-s} 
For any positive constants $\lambda,\gamma, \kappa, \rho, \varrho, \phi, T$ we have 
\begin{itemize}
\item[(i)]
\begin{equation*}
\sup_{0\leq t\leq T}|S_{4,j}(T-t)|<\infty \quad \textrm{and} \quad \sup_{0\leq t\leq T}|G_{j}(T-t)|<\infty \quad (j \in \{ 1,\ldots,4 \}),
\end{equation*}
\item[(ii)]
\begin{equation*}
1 \leq \inf_{0\leq t \leq T} S_{4,4}(T-t) \leq  \sup_{0\leq t \leq T}S_{4,4}(T-t) <\infty,
\end{equation*}
\item[(iii)]
\begin{equation*}
- \infty< \inf_{0\leq t \leq T} G_{3}(T-t) \leq  \sup_{0\leq t \leq T}G_{3}(T-t) <-1.
\end{equation*}
\end{itemize}
\end{lemma} 

\begin{proof} 
\emph{(i):} This follows directly from the explicit expressions in~\eqref{eq:S41}--\eqref{eq:G4}.

\emph{(ii):} It suffices to show that the continuously differentiable mapping $t \mapsto S_{4,4}(T-t)$ in~\eqref{eq:S44} is decreasing on $[0,T]$ with $S_{4,4}(0)=1$. To achieve this, it is convenient to introduce
\begin{equation}
\begin{aligned}
I_1(t) \triangleq & \, (\nu_1^2-\rho^2) (\nu_3^2 + \rho^2) \cosh(\nu_3(T-t)), \\
I_2(t) \triangleq & \, - (\nu_1^2 + \rho^2) (\nu_3^2-\rho^2) \cosh(\nu_1(T-t)), \\
I_3(t) \triangleq & \, 2\rho\nu_3 (\nu_1^2-\rho^2) \sinh( \nu_3(T-t)), \\
I_4(t) \triangleq & \, - 2\rho\nu_1 (\nu_3^2-\rho^2)  \sinh( \nu_1(T-t))
\end{aligned}
\end{equation}
for all $t \in [0,T]$ and to rewrite $S_{4,4}$ in~\eqref{eq:S44} as
\begin{equation} 
S_{4,4}(T-t) = \frac{1}{2\rho^2(\nu_1^2 -
             \nu_3^2)}\big(I_1(t) + I_2(t) + I_3(t) + I_4(t)\big). 
\end{equation}
We claim that 
\begin{equation} \label{eq:proof:lemma-g-s:claim}
\frac{d}{dt} S_{4,4}(T-t) = \frac{1}{2\rho^2(\nu_1^2 -
             \nu_3^2)}\left( I_1'(t) + I_2'(t) + I_3'(t) + I_4'(t) \right) \leq 0 
\end{equation}
for all $t \in [0,T]$. First, since $c_2 > 0$ in~\eqref{def:c1c2} we have that $\nu_1^2 - \nu_3^2=-c_2<0$ and hence
\begin{equation} \label{eq:proof:lemma-g-s:nus1}
\frac{1}{2\rho^2(\nu_1^2 -
             \nu_3^2)} <0. 
\end{equation}
Moreover, by virtue of Lemma~\ref{lem:aux1}, together with the fact that $\nu_1, \nu_3 <0$ in~\eqref{def:eigenvalues}, which also implies $\sinh(\nu_1(T-t)) \leq 0$ and $\sinh(\nu_3(T-t)) \leq 0$ for all $t \in [0,T]$, it follows that 
\begin{equation*}
\begin{aligned}
I_1'(t) & = -\nu_3 (\nu_3^2+\rho^2)(\nu_1^2-\rho^2) \sinh(\nu_3(T-t)) \geq 0, \\ 
I_2'(t) & = \nu_1(\nu_1^2+\rho^2)(\nu_3^2-\rho^2) \sinh(\nu_1(T-t)) \geq 0, \\ 
I_3'(t) & = -2\rho \nu_3^2(\nu_1^2-\rho^2) \cosh(\nu_3(T-t))  \geq \; 0,  \\ 
I_4'(t) &=  2\rho \nu_1^2(\nu_3^2-\rho^2) \cosh(\nu_1(T-t))  \geq 0 \\ 
\end{aligned}
\end{equation*}
and thus our claim in~\eqref{eq:proof:lemma-g-s:claim}. Finally, observe that $S_{4,4}(0)=1$.

\emph{(iii):} First, we emphasize the dependence of $G_3(\cdot)$ in~\eqref{eq:G3} on $\varrho$ by writing
$$
\begin{aligned}
G_{3}(T-t;\varrho)  
    \triangleq & \, G_{3}(T-t) \\
    = & \, \frac{1}{2\lambda\nu_1\nu_3 (\nu_1^2 - \nu_3^2)} \\
    & \Big( \nu_1\nu_3 \big( \gamma\kappa\rho  - 2\lambda (\nu_1^2 - \rho^2)\big) \cosh( \nu_1(T-t)) \Big. \nonumber \\
    & \hspace{6pt} - \nu_3 \big(\gamma\kappa\nu_1^2 +2\varrho (\nu_1^2-\rho^2) \big) \sinh(\nu_1(T-t)) \nonumber \\
    & \hspace{6pt} - \nu_1\nu_3 \big( \gamma\kappa\rho - 2 \lambda (\nu_3^2-\rho^2) \big) \cosh( \nu_3(T-t)) \nonumber \\
    & \hspace{6pt} \Big. + \nu_1 \big( \gamma\kappa\nu_3^2 + 2\varrho (\nu_3^2-\rho^2)\big) \sinh(\nu_3(T-t))
    \Big) \qquad (0 \leq t \leq T).
\end{aligned}
$$
Note that similar to~\eqref{eq:proof:lemma-g-s:nus1} above $\nu_1, \nu_3 <0$ implies
\begin{equation} \label{eq:proof:lemma-g-s:nus2}
\frac{1}{2\lambda\nu_1\nu_3 (\nu_1^2 - \nu_3^2)} < 0. 
\end{equation}
Moreover, using Lemma~\ref{lem:aux1} and having in mind that $\sinh(\nu_1(T-t)) \leq 0$ and $\sinh(\nu_3(T-t)) \leq 0$ for all $t \in [0,T]$, one can easily check that
\be \label{g3-r}
G_{3}(T-t;\varrho) \leq G_{3}(T-t;0) \qquad (0\leq t\leq T).
\ee
Indeed, using the definition of $G_3$, inequality~\eqref{g3-r} is equivalent to
\be
\nu_1 (\nu_3^2-\rho^2) \sinh(\nu_3(T-t)) \geq \nu_3 (\nu_1^2-\rho^2)
\sinh(\nu_1(T-t)) \qquad (0\leq t\leq T),
\ee
which holds true because the left-hand side is non-negative and
the right-hand side is non-positive.

Next, introducing
\begin{align*}
K_1(t) \triangleq & \; \nu_1\nu_3 \big( \gamma\kappa\rho  - 2\lambda (\nu_1^2 - \rho^2)\big) \cosh( \nu_1(T-t)), \\
K_2(t) \triangleq & \; - \gamma\kappa\nu_3 \nu_1^2  \sinh(\nu_1(T-t)), \\
K_3(t)\triangleq & \; - \nu_1\nu_3 \big( \gamma\kappa\rho - 2 \lambda (\nu_3^2-\rho^2) \big) \cosh( \nu_3(T-t)), \\
K_4(t) \triangleq & \; \gamma\kappa \nu_1\nu_3^2  \sinh(\nu_3(T-t))
\end{align*}
for all $t \in [0,T]$ allows us to write
\begin{equation*}
G_{3}(T-t;0) = \frac{1}{2\lambda\nu_1\nu_3 (\nu_1^2 -
             \nu_3^2)}\big(K_1(t) + K_2(t) + K_3(T) + K_4(t)\big) \quad (0 \leq t \leq T).
\end{equation*} 
Let us also define 
\be \label{g-3-t}
\wt G_{3}(T-t) \triangleq \frac{1}{2\lambda\nu_1\nu_3 (\nu_1^2 - \nu_3^2)}\big(K_1(t)  + K_3(t) \big). 
\ee
Since $|\nu_1|<|\nu_3|$ in~\eqref{def:eigenvalues} we obtain
\begin{equation} \label{g-4-t}
K_2(t) + K_4(t) =  -\gamma\kappa\nu_3 \nu_1^2  \sinh(\nu_1(T-t)) + \gamma\kappa \nu_1\nu_3^2  \sinh( \nu_3(T-t)) \geq 0 
\end{equation}
for all $t \in [0,T]$. In fact, note that $|\nu_1|<|\nu_3|$ implies
\be
-\nu_3 \sinh(-\nu_3(T-t)) \geq -\nu_1 \sinh(-\nu_1(T-t)),
\ee
which is equivalent to~\eqref{g-4-t}. Together with~\eqref{eq:proof:lemma-g-s:nus2} we can therefore conclude that
\be \label{tg3-r}
G_{3}(T-t;\varrho)  \leq G_{3}(T-t;0) \leq \wt G_{3}(T-t) \qquad (0\leq t\leq T). 
\ee
We will now argue that the continuously differentiable mapping $t \mapsto \wt G_3(T-t)$ in~\eqref{g-3-t} is increasing on $[0,T]$. The bounds on $G_3$ in (iii) will then follow from (\ref{tg3-r}) together with $\wt G_{3}(0)=-1$. To this end, simply observe that
\begin{equation*}
\begin{aligned}
K_1'(t) & =  -\nu_1^2 \nu_3 (\gamma\kappa\rho-2\lambda(\nu_1^2-\rho^2))\sinh(\nu_1(T-t)) \leq 0,  \\ 
K_3'(t) & = \nu_1\nu_3^2 \big( \gamma\kappa\rho - 2 \lambda (\nu_3^2-\rho^2) \big) \sinh(\nu_3(T-t)) \leq 0,
\end{aligned}
\end{equation*}
because of Lemma~\ref{lem:aux1}, (\ref{def:eigenvalues}) and the fact
that
\be \label{tg4-r}
\gamma\kappa\rho - 2 \lambda (\nu_3^2-\rho^2) = -\phi + \lambda \rho^2
- \sqrt{4\gamma\kappa\phi\rho + (\phi -
  \rho(\gamma\kappa+\lambda\rho))^2} \leq 0,
\ee
which
holds true due to a similar reasoning as in the proof of
Lemma~\ref{lem:aux1} above. More precisely, 
if $-\phi + \lambda \rho^2 \leq 0$, then~\eqref{tg4-r} holds
trivially. Otherwise, if  $-\phi + \lambda \rho^2 > 0$,
then~\eqref{tg4-r} is equivalent to
$2\gamma\kappa\phi\rho+2\lambda\rho^3\gamma\kappa+\rho^2\gamma^2\kappa^2
>0$, which is satisfied. Consequently, recalling~\eqref{eq:proof:lemma-g-s:nus2} we obtain
in~\eqref{g-3-t} that
$$
\frac{d}{dt}\wt G_3(T-t) \geq 0 \quad (0\leq t\leq T) 
$$
as desired.
\end{proof} 




\end{document}